\documentclass[11pt,a4paper]{article}

\usepackage[utf8]{inputenc} 
\usepackage[T1]{fontenc}    %
\usepackage[english]{babel} 
\usepackage[final]{microtype} 

\usepackage{amssymb,amsmath,mathtools} 
\usepackage{amsthm} 

\usepackage{authblk}

\usepackage{xcolor}
\usepackage{bbm}
\usepackage{braket}
\usepackage{todonotes}
\usepackage{subcaption}

\usepackage[hmargin=0.12\paperwidth,vmargin=0.16\paperwidth,bindingoffset=0cm]%
{geometry} 

\numberwithin{equation}{section} 

\theoremstyle{plain}
  
  \newtheorem{prop}{Proposition}
  
  \newtheorem{cor}{Corollary}
\theoremstyle{definition}
  
  \newtheorem{remark}{Remark}
  
  \numberwithin{prop}{section}
   \numberwithin{cor}{section}
   \numberwithin{remark}{section}

\begin{document}
\date{}

\title{Properties of the one-component Coulomb gas on a sphere with two macroscopic external charges}


\author[1]{Sung-Soo Byun}
\affil[1]{\small Department of Mathematical Sciences and Research Institute of Mathematics, Seoul National University, Seoul 151-747, Republic of Korea}

\author[2]{Peter J. Forrester}
\affil[2]{School of Mathematical and Statistics, The University of Melbourne, Victoria 3010, Australia}

\author[2,3]{Sampad Lahiry}
\affil[3]{
Department of Mathematics, Katholieke Universiteit Leuven, Celestijnenlaan
200 B bus 2400, Leuven, 3001, Belgium}

\affil[  ]{sungsoobyun@snu.ac.kr; \,pjforr@unimelb.edu.au; \, sampad.lahiry@kuleuven.be}


\maketitle


\begin{abstract}
The one-component Coulomb gas on the sphere, consisting on $N$ unit charges interacting via a logarithmic potential, and in the presence of two external charges each of strength proportional to $N$, is considered.
There are two spherical caps naturally associated with the external charges, giving rise to two distinct phases depending on them not overlapping (post-critical) or overlapping (pre-critical). The equilibrium measure in the post-critical phase is known from earlier work. We determine the equilibrium measure in the pre-critical phase using a particular conformal map, with the parameters therein specified in terms of a root of a certain fourth order polynomial. This is used to determine the exact form of the electrostatic energy for the pre-critical phase. Using a duality relation from random matrix theory, the partition function for the Coulomb gas at the inverse temperature $\beta = 2$ can be expanded for large $N$ in the post-critical phase, and in a scaling region of the post and pre-critical boundary. For the pre-critical phase, the duality identity implies a relation between two electrostatic energies, one for the present sphere system, and the other for a certain constrained log-gas relating to the Jacobi unitary ensemble.
\end{abstract}


\section{Introduction}
\subsection{Outline of the model system and phase transition effect}\label{S1.1}
A defining aspect  of a many body Coulomb system is that the pair interactions are specified by a potential coming from the solution of the appropriate Poisson equation. This will depend on the space dimension and also geometry. For example, in two dimensions, the solution of the Poisson equation $\nabla_{\mathbf r}^2 \Phi(\mathbf r, \mathbf r') = - 2 \pi \delta(\mathbf r - \mathbf r')$ is the logarithmic potential $ \Phi(\mathbf r, \mathbf r')  = - \log |\mathbf r - \mathbf r'|$. Here it is assumed that the two-dimensional space is flat, and is unique only up to an additive constant. Other often considered circumstances are when the two-dimensional space has periodic boundary conditions in one direction ($x$-direction say, of period $L$), or when it is
 the surface of a sphere of radius $R_{\rm s}$. In the former situation, the solution of the Poisson equation is
 $- \log | \sin(\pi(x-x' + i (y - y'))/L)|$, unique up to an additive constant.
  However, for any smooth $\Phi(\mathbf r)$ on the surface of the
 sphere, $\Omega$ say, $\int_\Omega \Phi(\mathbf r) \, d \mathbf r = 0$. Thus the above written Poisson
 equation does not have a solution on the sphere but rather must be modified to be charge neutral in the
 sense that it integrates to zero, which can be achieved by adding $1/(2 R_{\rm s}^2)$. The solution of this
 modified Poisson  equation can again be shown to be given by $- \log   |\mathbf r - \mathbf r'|$, up to
 an additive constant \cite[\S 15.6.1]{Fo10}.
 One notes that for $\mathbf r, \mathbf r'$ on the surface of the sphere, $ |\mathbf r - \mathbf r'|$ corresponds
 to the chord length. With $(\theta, \phi)$ the usual polar and azimuthal angles determining $\mathbf r$ on the
 sphere, a more convenient form is \cite[Eq.~(15.109)]{Fo10}
\begin{equation}\label{1A}
 \Phi((\theta, \phi), (\theta',\phi'))  = - \log (2R_{\rm s}|u'v-uv'|),
 \end{equation}
  where use has been made of the Cayley-Klein parameters
  \begin{equation}\label{1B}
  u = \cos(\theta/2) e^{i \phi/2}, \qquad  v = -i \sin(\theta/2) e^{-i \phi/2}.
   \end{equation}
   
The one-component plasma model of a Coulomb system consists of a smeared out neutralising background charge $- \rho_{\rm b}(\mathbf r)$, and $N$ mobile point particles with unit charge. The mobile point particles interact with the pairwise
   potential coming from the solution of the Poisson equation, and also a one-body potential due to smeared out
   background, $V(\mathbf r') =  \int_\Omega \log | \mathbf r - \mathbf r'| \rho_{\rm b}(\mathbf r') \, d \mathbf r'$. 
   On the surface of the sphere, a natural choice for the neutralising background density is the constant
   $ - \rho_{\rm b}(\mathbf r) = - N/(4 \pi R_{\rm s}^2)$. Then, by rotational invariance, $V(\mathbf r')$ is a constant.
   Computing the total energy of the system, consisting of the background-background, background-particle, and particle-particle interactions, allows for the computation of the Boltzmann factor as \cite[Eq.~(15.110)]{Fo10}
 \begin{equation}\label{1C}
 \Big ( {1 \over 2 R_{\rm s}} \Big )^{N \beta/2} e^{\beta N^2/4} \prod_{1 \le j < k \le N} | u_k v_j - u_j v_k|^\beta.
 \end{equation}
 
 In the present work, guided by considerations in random random matrix theory
 and also in the light of the recent work \cite{BSY24}, we are lead to consider the
 generalisation of (\ref{1C})
\begin{equation}\label{1D}  
|u_w|^{\beta Q_0 Q_1 N^2}  \Big ( {1 \over 2 R_{\rm s}} \Big )^{N \beta/2} \mathcal K_N^{\rm post}
 \prod_{l=1}^N | u_l|^{\beta Q_0N} |u_l v_w - u_w v_l |^{\beta Q_1 N}
 \prod_{1 \le j < k \le N} | u_k v_j - u_j v_k|^\beta,
 \end{equation}
 where the Cayley-Klein parameters $(u_w,v_w)$ are given and
 \begin{align}\label{C2}
 \mathcal K_N^{\rm post} :=  & \exp \bigg( - {\beta N^2 \over 4} \Big (
 - (1 + Q_0 + Q_1) + 2 (1 + Q_0 + Q_1) \log {1 \over  1 + Q_0 + Q_1}  \nonumber \\
 & \quad + (1 + Q_0)^2 \log (1 + Q_0) + (1 + Q_1)^2 \log (1 + Q_1) - Q_1^2 \log Q_1 - Q_0^2 \log Q_0 \Big ) \bigg ).
 \end{align} 
 (The reason for the superscript in $\mathcal K_N^{\rm post} $ will become apparent later.)
In fact, with $(u_w,v_w) = (1,0)$ this has appeared earlier in \cite[with $Q_0 = Q$, $Q_1 = q$]{FF11}, where it was also motivated by considerations in random matrix theory, albeit different from the one to be considered here. 
It is the Coulomb system interpretation of (\ref{1D}) as the Boltzmann factor of a 
(generalised) one-component plasma on the sphere that provides the expression for
$\mathcal K_N^{\rm post}$. Thus at the location on the sphere with Cayley-Klein parameters $(u,v)=(0,1)$ (this is the south pole)
there is a charge of strength $Q_0N$, while at the point with Cayley-Klein parameters $(u,v)= (u_w,v_w)$
there is a charge of strength $Q_1N$. Let there be $N$ mobile unit charges.
Since the strength of these charges is chosen to be proportional to $N$, they can be viewed as macroscopic external charges.
Introducing a uniform neutralising background
 charge density $- \rho_{\rm b} = -N(1+Q_0+Q_1)/(4 \pi R_{\rm s}^2)$ and computing the total potential energy in the system
 then gives (\ref{1D}) with $\mathcal K_N^{\rm post}$ given by (\ref{C2}); see \cite{FF11} for details of the required working.
 
 The specific question to be addressed is the computation of the large $N,R_{\rm s}$, $\rho_{\rm b}$ fixed (i.e.~the thermodynamic
 limit)  expansion of (\ref{1D}) as a function of $w$. In the case of (\ref{1D}) being suggested by random matrix
 theory ($\beta = 2$), there are structures present which permit  a detailed asymptotic analysis. A phase transition
 effect is found depending on $w, Q_0, Q_1$,
 
 \begin{prop}\label{P1.1}
 Consider the partition function associated with (\ref{1D}) restricted to $\beta = 2$,
 \begin{multline}\label{1E}  
 Z_N^{\rm s}(w;Q_0,Q_1,R_{\rm s}) :=  |u_w|^{\beta Q_0 Q_1 N^2}  \Big ( {1 \over 2 R_{\rm s}} \Big )^{N \beta/2}  {\mathcal K_N^{\rm post} \over N!}  \\
 \times \int_\Omega d \Omega_1 \cdots \int_\Omega d \Omega_N \,
 \prod_{l=1}^N | u_l|^{\beta Q_0N} |u_l v_w - u_w v_l |^{\beta Q_1 N}
 \prod_{1 \le j < k \le N} | u_k v_j - u_j v_k|^\beta \Big |_{\beta = 2},
 \end{multline}
where $\Omega$ denotes the surface of the sphere and $d \Omega_l = R_{\rm s} \sin \theta_l d \theta_l d \phi$. Require that
 \begin{equation}\label{1F}
| u_w |^2 >   {1 \over (\gamma_1 + \gamma_2 + 2)^2} \Big ( \sqrt{(\gamma_1 + 1) (\gamma_1 + \gamma_2 + 1)} + \sqrt{\gamma_2 + 1} \Big )^2 
\Big |_{\substack{\gamma_1 = -1+ Q_0/Q_1 \\  \gamma_2 =1/Q_1}},
\end{equation}
and fix $ \rho_{\rm b} = N(1+Q_0+Q_1)/(4 \pi R_{\rm s}^2)$. We have the large $N$ expansion, independent of $(u_w,v_w)$ to all inverse powers of $N$,
 \begin{equation}\label{1G}
  \log  Z_N^{\rm s}(w;Q_0,Q_1,R_{\rm s}) = -{N \over 2} \log {\rho_{\rm b} \over 2 \pi^2} + {1 \over 12} \log {Q_0 \over 1 + Q_0} + {1 \over 12} \log {Q_1 \over 1 + Q_1} + \sum_{k=1}^\infty {a_k \over N^k} + {\rm O}( e^{-\epsilon N}) ,
 \end{equation}
for coefficients $\{a_k\}$ expressible in terms of the Bernoulli numbers and for some $\epsilon>0$.
 This expansion breaks down for $| u_w|^2$ small enough that (\ref{1F}) no longer holds, with the leading term then being proportional to $N^2$, and
 the expansion now dependent on $(u_w,v_w)$.
 \end{prop}

There is a different interpretation of the background density as introduced below (\ref{C2}) which sheds some light on the phase transition effect, and the critical value as implied by (\ref{1F}). Thus require that a spherical cap of area $|\Omega_0| Q_0/(Q_0+Q_1+1)$ (resp., $|\Omega_0| Q_1/(Q_0+Q_1+1)$) centred about the south pole (resp., about $(u_w,v_w)$) be free of background charge, where $|\Omega| = 4 \pi R_{\rm s}^2$ is the surface area of the sphere. These fractions are chosen as to be the proportion of the total charge that is contained at the respective points.
 Assuming that the two spherical caps do not overlap, the remaining surface area of the sphere is
 $|\Omega| /(Q_0+Q_1+1)$. In this remaining surface area, impose a neutralising background charge density with respect to
 the $N$ mobile charge. This again leads to the value $- \rho_{\rm b} = -N(1+Q_0+Q_1)/(4 \pi R_{\rm s}^2)$, but supported in the restricted
 domain, $\Omega^{\rm d}$ say, only. The calculation specified in \cite{FF11} in the case $(u_w,v_w) = (1,0)$ again applies and
 then gives precisely the same Boltzmann factor (\ref{1D}). This calculation implies that for $ \mathbf r \in \Omega^{\rm d}$,
 \begin{equation}\label{1H} 
 Q_1 \log |   \mathbf r -  \mathbf r_w | + Q_0 \log  |   \mathbf r -  \mathbf r_0 |  + \rho_{\rm b} \int_{\Omega^{\rm d}}  \log |   \mathbf r -  \mathbf r' | \, d  \Omega'  = C,
 \end{equation} 
 where $\mathbf r_w,  \mathbf r_0$ are the points on the sphere corresponding to $(u_w,v_w)$ and the south pole respectively, and $C$ is
 a constant. According to potential theory \cite{ST97},
 this is a requirement for $ \rho_{\rm b} \mathbbm 1_{ \mathbf r \in \Omega^{\rm d}}$ to be the minimiser with respect to sub-domains of the sphere $\Lambda$ of the logarithmic energy functional $\mathcal E_N(\Lambda)$,
 \begin{equation}\label{E1y}
 \mathcal E_N[\Lambda] := -\int_{\Lambda}   \Big ( Q_1 \log |\mathbf r'  -  \mathbf r_w | + Q_0 \log  |   \mathbf r' -  \mathbf r_0 | \Big ) d \Omega' - {\rho_{\rm b} \over 2} \int_{\Lambda} d \Omega' \,
 \int_{\Lambda} d \Omega \,   \log |   \mathbf r -  \mathbf r' |
  \end{equation}
 associated
 with the external potential corresponding to the two external charges. An analogous result holds true for any number of external charges,
 provided that the associated spherical caps do not overlap \cite{BDSW18}.
 However, this simple specification of $ \Omega^{\rm d}$ --- the so-called droplet --- breaks down as soon as the spherical caps overlap
 \cite{CK19,LD21,CK22}. 
 As verified in Appendix A, the spherical caps associated with the external charges do not overlap if and only if the condition (\ref{1F}) holds.
 As a consequence the computation of the constant terms (\ref{1G}) in the Boltzmann factor (\ref{1D}) is no longer applicable. We would like to
 specify its replacement, and similarly the replacement of (\ref{1G}). This task in turn relies on detailed knowledge of the droplet in the setting that the spherical caps overlap, which in the terminology of \cite{BBLM15} is the pre-critical phase; the case when the spherical caps do not overlap is referred to as the post-critical phase. We undertake the required working herein, and are able to specify the replacement of (\ref{C2}).

\subsection{Relationship to earlier work} 
  The work  \cite{BSY24} considers a planar analogue of (\ref{1D}), for which the Boltzmann factor
 is equal to
 \begin{equation}\label{1Dx}  
\mathcal K_N \prod_{l=1}^N e^{-\beta |z_l|^2/2} |w - z_l|^{\beta Q N} \prod_{1 \le j < k \le N} | z_k - z_j|^\beta,
 \end{equation}
 where $z_l = x_l + i y_l$ with $(x_l,y_l)$ the Cartesian coordinates in the plane and
  \begin{equation}\label{1Dy}  
\mathcal K_N = e^{- (\beta N Q/2) ( |w|^2 + R^2 \log R^2 - R^2/2) - \beta N^2 ((1/4) \log R^2 - 3/8)} \Big |_{R^2 = N (1 + Q)};
 \end{equation}
 in relation to the computation of $ \mathcal K_N $, see \cite[Exercises 1.4 q.2]{Fo10}.
As a Coulomb system, this results from a one-component system
 with a smeared out background charge density $- \rho_{\rm b} = - {1 \over \pi}$ in a disk about the
 origin of radius $r_0$, equal to the square root of $N(1 + Q)$. In addition to the $N$ mobile particles of
 unit charge, there is a fixed particle at $w$ of charge $QN$. In keeping with the text in the paragraph including
 (\ref{1H}), associating with the fixed charge at $w$ a disk of radius equal to the square root of $QN$, provided this lies
 entirely inside the disk of radius $r_0$ centred at the origin, $\mathbb D_{r_0}$ say, the uniform background can be taken to be $D^{\rm d}(w)$, this denoting  $\mathbb  D_{r_0}$, with the disk about  $w$ removed. With this interpretation,
  $D^{\rm d}(w)$ specifies the droplet.  Scaling $w = \sqrt{N} a$, a simple calculation shows that the condition that the disk of radius
  $\sqrt{QN}$ centred at $\sqrt{N} a$ be strictly contained in $D_{r_0}$ is
 \begin{equation}\label{1Cp}
   |a| < \sqrt{1+Q} - \sqrt{Q}.
 \end{equation}   
 
It was observed in \cite{BSY24} (see too \cite{DS22}) that for the coupling $\beta = 2$ a duality identity from random matrix theory  facilitates the large $N$ analysis of the partition function corresponding to (\ref{1Dx}). The first point to note is the random matrix interpretation of (\ref{1Dx}) in the case $\beta = 2$. On this one recalls that the complex Ginibre ensemble (denoted GinUE) consists of $N \times N$ matrices with independent standard complex Gaussian
 entries. The eigenvalue joint probability density function has the explicit functional form proportional to
 $\prod_{l=1}^N e^{- |z_l|^2} \prod_{1 \le j < k \le N} | z_k - z_j|^2$; 
see  \cite[Prop.~2.1]{BF24}. Hence for $\beta = 2$ we see that (\ref{1Dx}) is proportional to
$\langle | \det (w \mathbb I_N - X) |^{2 Q N} \rangle_{X \in {\rm GinUE}}$. On the other hand, this quantity for $QN$ a positive integer obeys the duality identity \cite{FR09,Fo24+}
   \begin{equation}\label{1Dz}  
 \langle | \det (w \mathbb I_N - X) |^{2 Q N} \rangle_{X \in {\rm GinUE}} = \Big \langle  \prod_{l=1}^{QN} (|w|^2 + t_l)^N \Big \rangle_{\mathbf t \in {\rm LUE}_{QN,0}}.
  \end{equation}  
Here LUE${}_{n,a}$ is the random matrix ensemble --- referred to as the Laguerre unitary ensemble --- with eigenvalue probability density function (PDF) proportional to
   \begin{equation}\label{1Ex}  
   \prod_{l=1}^n \lambda_l^a e^{- \lambda_l} \prod_{1 \le j < k \le n} (\lambda_k - \lambda_j)^2,
   \end{equation}   
   supported on $\lambda_l > 0$. Integrating this PDF from $s$ to infinity gives the probability $E^{ {\rm LUE}_{n,a}}(0;(0,s))$ that there are no eigenvalues in $(0,s)$, which is generally referred to as a gap probability.
   A simple change of variables shows that the RHS of (\ref{1Dz})
   can be written in terms of this gap probability to give
\begin{equation}\label{1Ey} 
 \langle | \det (w \mathbb I_N - X) |^{2 Q N} \rangle_{X \in {\rm GinUE}} =  C_{N,QN}  e^{ QN |w|^2} E^{ {\rm LUE}_{QN,N}}(0;(0,s)),
 \end{equation} 
 where $C_{N,QN}$ is a ratio of normalisations associated with (\ref{1Ex}), which are known in terms of products of gamma functions; see e.g.~\cite[\S 4.7.1]{Fo10}.
 While the large $N$ form of the Ginibre ensemble average on the LHS is difficult to analyse directly,
the gap probability on the RHS can be analysed
   using various tools, the simplest being a concentration of measure argument (see \S \ref{S3.3} below).
   One consequence is that in the regime \cite[equivalent to Th.~2.2 post-critical case]{BSY24}
 when   the droplet is given by   ${D}^{\rm d}(w)$, the partition function corresponding to (\ref{1Dx}) with $\beta = 2$, i.e.~the integral over $\mathbb R^N$ of
    (\ref{1Dx}) with this $\beta$ multiplied by $1/N!$, $ Z_N^{\rm d}(w;Q) $ say, has with $w = a \sqrt{N}$ the large $N$ form
  \begin{equation}\label{2Ex}  
\log  Z_N^{\rm d}(a \sqrt{N};Q) =    -{N \over 2} \log {\rho_{\rm b} \over 2 \pi^2} +  {1 \over 12} \log {Q \over 1 + Q} +  \sum_{l=1}^\infty {c_l \over N^l} + {\rm O}(e^{-\epsilon N}),
 \end{equation}
for some $\epsilon > 0$, where coefficients $\{c_l\}$  in the $1/N$ expansion are independent of $a$ (these terms, which involve the Bernoulli numbers, are given explicitly in \cite[Eq.~(2.8), after adjusting by the addition of the
$1/N$ expansion of $-\log N!$ to concur with our definition of the partition function]{BSY24}).

Define $|a_c|$ as the RHS of (\ref{1Cp}), and on the LHS of (\ref{2Ex}) with $a$ positive real for convenience, expand
 \begin{equation}\label{2F}  
 a  = a_c - {( \sqrt{1+Q} - \sqrt{Q})^{1/3} \over 2 Q^{1/6} (Q + 1)^{1/6}} {s \over N^{2/3}} =: a(s;Q,N).
  \end{equation}
Introduce from random matrix theory \cite{TW94a,FW15} the $\beta = 2$ soft edge scaled probability of no eigenvalue in $(t,\infty)$,
 \begin{equation}\label{2F1}    
 E^{ \rm soft}_2(0;(t,\infty))  = \exp \bigg ( - \int_{t}^\infty (x - t) \mathbf q(x) \, dx \bigg ),
   \end{equation}
   where $ \mathbf q$ is the Hastings-McLeod solution to the Painlev\'e II equation
\begin{equation}\label{2F2}  
\mathbf q''(x) = x \mathbf q(x) + 2   ( \mathbf q(x) )^3,  \quad \mathbf q(x) \mathop{\sim}_{x \to \infty} {\rm Ai}(x).
 \end{equation}
 Then we have from \cite[Proposition 2.5]{BSY24} that for $\beta = 2$
 \begin{equation}\label{2E}  
\log  Z_N^{\rm d}(a(s;Q,N) \sqrt{N};Q) =    -{N \over 2} \log {\rho_{\rm b} \over 2 \pi^2} +  {1 \over 12} \log {Q \over 1 + Q} + \log  E^{ \rm soft}_2(0;(Q^{-2/3}s,\infty)) 
+ {\rm O}\Big ( {1 \over N^{2/3}} \Big ).
 \end{equation} 
 
 With $w = \sqrt{N}a$ the analogue of (\ref{1H}) for (\ref{1Dx}) is the equation
 \begin{equation}\label{1Ha} 
 - {1 \over 2} |z|^2  + Q \log  |   z -  a |  + \rho_{\rm b} \int_{D^{\rm d}}  \log |   z -  z' | \, d  \mathbf r'  = C,
 \end{equation}  
 valid for all $z \in D^{\rm d}$ with $C$ some constant. Here $\rho_{\rm b} = {1 \over \pi}$ and the task is to determine
 the domain $D^{\rm d}$ such that (\ref{1Ha}) holds, this specifying the droplet. Under the assumption (\ref{1Cp}), the solution of this equation is
 $D^{\rm d} = {1 \over \sqrt{N}} D^{\rm d}(w)$ where $D^{\rm d}(w)$ is specified in the paragraph including (\ref{1Dy}).
 With the inequality of (\ref{1Cp}) in the opposite direction, and thus the disk about $w$ not entirely inside $\mathbb D_{r_0}$, 
 we have from \cite[Section 2]{BBLM15} (see also \cite[\S 1.3]{BSY24}) that the boundary of $D^{\rm d} $ is the image of the
 unit circle under the conformal map
  \begin{equation}\label{1Hb} 
  f(z) = r z - {\kappa \over z - q} - {\kappa \over q}, \quad r = {1 + (a q)^2 \over 2 a q}, \quad \kappa  = {(1 - q^2) (1 - a^2 q^2) \over 2 a q}.
   \end{equation} 
   Here $q=q(a)$ satisfies $f(1/q) = a$ and can be determined as the unique solution of the polynomial equation
   \begin{equation}\label{1Hc} 
   q^6 - \Big ( {a^2 + 4 Q + 2 \over 2 a^2} \Big ) q^4 + {1 \over 2 a^4} = 0  
   \end{equation} 
   satisfying $0<q<1$ and $\kappa > 0$.  Moreover, results of \cite[Th.~2.1]{BSY24} give that the appropriate modification of
 (\ref{1Dx}) is the scaled Boltzmann factor  
    \begin{equation}\label{1DxA}  
\tilde{\mathcal K}_N \prod_{l=1}^N e^{-N \beta |z_l|^2/2} |a - z_l|^{\beta Q N} \prod_{1 \le j < k \le N} | z_k - z_j|^\beta,
 \end{equation}
 where $z_l = x_l + i y_l$ with $(x_l,y_l)$ the Cartesian coordinates in the plane of the $l$-th particle and
  \begin{multline}\label{1DyA}  
\tilde{\mathcal K}_N = \exp \bigg ( { \beta N^2 \over 2} \Big ( {3 \over 8} + {a^2 \over 8} + {3 \over 8 a^2 q^4} - {5 \over 8 q^2} +
\Big ( {3 \over 4} + {a^2 \over 8} \Big )(aq)^2 - {3 (aq)^4 \over 8} \\
+ \log (2 a q) + 2Q \log(2 a q^2) + \log {(1 + (aq)^2 - 2 a^2 q^4)^{Q^2} \over (1 + (aq)^2)^{(Q+1)^2}}  \Big )  \bigg ).
\end{multline}
Denoting the corresponding partition function by $ \tilde{Z}_N^{\rm d}(a;Q)$, this obtained from (\ref{1DxA}) by integrating over
$\mathbb C^N$ and  multiplying by $1/N!$, we have from \cite[Th.~2.2]{BSY24} that for large $N$
 \begin{equation}\label{2EA}  
\log   \tilde{Z}_N^{\rm d}(a;Q) =    -{N \over 2} \log {\rho_{\rm b} \over 2 \pi^2} - {1 \over 12} \log N + \zeta'(-1) + \mathcal F^{\rm pre}(a,c) 
+ {\rm O}\Big ( {1 \over N} \Big ),
 \end{equation} 
where $\rho_{\rm b} = {1 \over \pi}$ and
 \begin{equation}\label{2EB}  
 \mathcal F^{\rm pre}(a,c) = {1 \over 24} \log \bigg ( {(1  + (aq)^2 - 2 a^2 q^4)^{4} \over  (1 + (aq)^2)^4 (1 - q^2)^3 (1 - a^4 q^6)}\bigg ).
 \end{equation}  

 There is another line of previous work relevant to our study. Thus as referenced below (\ref{E1y}), these works consider the characterisation of the domain $\Omega^{\rm d}$ which minimise (\ref{E1y}). Specifically, in  \cite{CK19} in the case $Q_0 = Q_1$ and with $\mathbf r_w$ and $\mathbf r_0$ positioned symmetrically about the south pole, it is found that the droplet in the case that the inequality (\ref{1F}) is violated is
 (i.e.~in the pre-crtical phase), after stereographic projection onto the plane, a
particular ellipse. Moreover, results from \cite{CC03,EM07} imply that for general $Q_0 \ne Q_1$, with the charge $Q_0 N$ positioned at the south pole, $\Omega^{\rm d}$ can be parametrised in terms of conformal map from the interior of the unit circle to the exterior of the droplet which is a simple rational function with two distinct poles and one zero; see (\ref{B1t}) below. Crucial to our study is the use of the conformal map characterisation of the droplet in the pre-critical phase to compute the corresponding energy (\ref{E1y}).
On this we are guided by the earlier works \cite{BBLM15} and \cite{BSY24}.

Beyond the papers of the above discussion, a number of works have appeared in the literature in recent years addressing problems on the  Coulomb gas which relate to potential theory and phase transitions in ways which complement the present work. As an incomplete list we reference \cite{BM15,WW19,ABK21,AKS21,ACC23,BKS23,LY23,KLY23,Se24,By24,KKL24}. 
 
 \subsection{Summary of results and outline}
In \S \ref{Se2} we take up the task of specifying the droplet in the pre-critical phase, when the corresponding spherical caps overlap.  As mentioned above, this specification is via a conformal map.
In the case $Q_0 =Q_1$, with the two charges symmetrically placed about the south pole,  we show 
in Proposition \ref{P2.2x} how the results of \cite{CK19} can be reclaimed using properties of a particular functional form for the Stieltjes transform of the droplet measure.
The required strategy to determine the parameters in the conformal map gives some guidance  into how to proceed in the more complicated case of general $Q_0 \ne Q_1$, which we study with the charge of strength $Q_0N$ placed at the south pole. Here one of the three parameters of the conformal map, denoted $\alpha$, determines the other two; this parameter is given implicitly in terms of $Q_0,Q_,w$ as
as the smallest positive root of the fourth order polynomial (\ref{2.62}) with the other parameters then determined according to Proposition \ref{P2.7}. We recall a similar feature of the conformal map (\ref{1Hb}) for the disk model with a macroscopic point charge insertion, with the parameter $q$ being given as a particular root of the sixth order polynomial equation (\ref{1Hc}). At the end of \S  \ref{S2.3u} it is shown, by  scaling the variables so that radius of the sphere is effectively taken to infinity simultaneous with $Q_0$, how the characterisation of $f(z)$ in
(\ref{1Hb}) and (\ref{1Hc}) can be reclaimed from the conformal map of the present work.
The determinisation of the electrostatic energy, i.e.~the specification of (\ref{E1y}) for the equilibrium measure, is carried out in
\S \ref{S2.4r}. For this, the analytic properties of the conformal map are used in an essential way. As a consequence the replacement of the constant (independent of the particle coordinates) factor $\mathcal K^{\rm post}$ in (\ref{1D}), to be denoted $\mathcal K^{\rm pre}$ is determined according to (\ref{2.8c+}), Proposition \ref{P8x} and Proposition \ref{P8y}.
We mention that, although not stated explicitly in the text, our results on the equilibrium measure, combined with the standard convergence of empirical measures for planar symplectic ensembles and Coulomb gases with Neumann boundary conditions (see e.g. \cite{BC12}), immediately gives rise to the analogous results for the symplectic counterparts of the complex spherical ensembles \cite{May13,MP17,BF23a,BP24}, with the proviso that the charges are on the real axis. Similarly, the evaluation of energies also applies to the symplectic ensembles.

The analogue of the random matrix theory duality relation (\ref{1Ey}) as relevant to the present Coulomb gas on a sphere with two macroscopic external charges is introduced in \S \ref{S3.2}. In \S \ref{S3.2} it is
shown how this permits a quick proof of Proposition \ref{P1.1}, as well as its analogue for a scaling limit of the boundary between the post and pre-critical regimes. In \S \ref{S4t} it is shown how, in the pre-critical regime, the duality leads to an identity between electrostatic energies, one for the two-dimensional sphere system, and the other for a certain constrained one-dimensional log-gas coming from the Jacobi unitary ensemble. In light of the result (\ref{2EA}) for the disk Coulomb gas with a macroscopic charge deduced using the duality (\ref{1Ey}), it may seem that a further natural application of the sphere Coulomb gas duality is to compute higher order terms in the large $N$ expansion in the pre-critical phase. As discussed in Appendix C, this task can in fact
be deferred to a later study involving a Coulomb gas analogue of the averaged characteristic polynomial for truncated unitary matrices with Haar measure, which from a related viewpoint is subject to present attention \cite{BCMS25,BCMS25a}.

 \section{Equivalent planar Coulomb gas and the equilibrium measure}\label{Se2}
  \subsection{Stereographic projection}\label{S2.1}

We would like to change variables in (\ref{1D}) from the  Riemann sphere to the complex plane using 
\begin{equation}\label{SP}
e^{i \phi_l} \tan {\theta_l \over 2} = z_l,
\end{equation}
corresponding to a stereographic projection from the south pole to a plane at the north pole of a sphere of radius $R=1/2$ (see e.g.~\cite[Figure 15.2]{Fo10}).  This task, with $\beta = 2$, $Q_1=0$, $\beta Q_1 N /2 = K$ has been done previously in \cite[Eq.~(2.23)]{FF11},
and with $Q_0=Q_1=0$ and general $\beta$ in \cite[Eq.~(15.128)]{Fo10}. We require that a pair of
points on the sphere $\mathbf a$ and $\mathbf a'$ with 
Cayley-Klein parameters $(u,v)$ and $(u',v')$
 relate to a pair of points $z$ and $z'$ in the complex plane according to
\begin{equation}\label{SP1}
|\mathbf a - \mathbf a'| = | u' v - u v'| = \cos {\theta \over 2} | z - z' |  \cos {\theta' \over 2}.
\end{equation}

 \begin{prop}\label{P2.1}
 With  $z_l = x_l + i y_l$, write $d \mathbf r_l = dx_l dy_l$, and denote by $d \Omega$ the area differential on the sphere of radius $1/2$.
 Map from the sphere to the plane via (\ref{SP}), and let
 $(u_w,v_w)$ be the Cayley-Klein parameters on the sphere relating to the point $w$ in the plane.
 We have
 \begin{multline}\label{13}
  |u_w|^{2Kr}  \prod_{l=1}^N  |u_l|^{2K} |u_w v_l - u_l v_w|^{2r}
  \prod_{1 \le j < k \le N} | u_k v_j - u_j v_k |^2 \, d\Omega_1 \cdots d \Omega_N \\
  =
{1 \over  (1 + |w|^2)^{r(K+N)} }
\prod_{l=1}^N { | w -  z_l |^{2r}  \over   (1 + | z_l|^2)^{K+r+N+1} } \prod_{1 \le j < k \le N} |z_k - z_j|^2 \, d\mathbf r_1 \cdots d \mathbf r_N.
  \end{multline}
  \end{prop} 
  
 \begin{proof}
With
\begin{equation}
w =  e^{i \phi_w} \tan {\theta_w \over 2}
\end{equation}
we see from (\ref{SP1}) that 
\begin{equation}
 \prod_{l=1}^N  | w -  z_l |^{2 r} = {1 \over |u_w|^{2rN} } \prod_{l=1}^N {1 \over |u_l|^{2r}} |u_w v_l - u_l v_w|^{2r}.
 \end{equation}
Combining this with   \cite[Eq.~(2.23) with $K \mapsto K + r$]{FF11} shows
\begin{multline}
  {1 \over |u_w|^{2rN} } \prod_{l=1}^N  |u_l|^{2K} |u_w v_l - u_l v_w|^{2r}
  \prod_{1 \le j < k \le N} | u_k v_j - u_j v_k |^2 \, d\Omega_1 \cdots d \Omega_N \\
 =  \prod_{l=1}^N { | w -  z_l |^{2r}  \over (1 + | z_l|^2)^{K+r+N+1} }  \prod_{1 \le j < k \le N} |z_k - z_j|^2 \, d\mathbf r_1 \cdots d \mathbf r_N \\
  \end{multline}
  Now noting
  $$
|u_w|^2 = {1 \over 1 + |w|^2}
$$
gives (\ref{13}).  
\end{proof}

By rotation invariance of the sphere, without loss of generality we can choose $\phi_w = 0$ so that $w$ is positive real, a condition that will
be assumed henceforth.

Following \cite{CK19}, particularly in the case $Q_0 = Q_1$, we will see that there is some advantage in rotating the sphere so that the two charges are located
symmetrically with respect to the south pole. Specifically, we locate the charge $Q_0 N$ at the point on the sphere with azimuthal angle $\phi = \pi$ and
polar angle ${1 \over 2} (\pi + \theta_w)$, and the charge $Q_1 N$ at the point on the sphere with azimuthal angle $\phi = 0$ and polar angle ${1 \over 2} (\pi + \theta_w)$. Denoting the Cayley-Klein parameters of the latter by $(u_{w_{\rm s}},v_{w_{\rm s}})$, then  the Cayley-Klein parameters
of the former are $(iu_{w_{\rm s}},-iv_{w_{\rm s}})$. The effect of this rotation in (\ref{1D}) is to replace the factor $|u_l|^{\beta Q_0 N} | u_l v_{w_{\rm s}} - u_{w_{\rm s}}  v_l |^{\beta Q_1 N}$
by 
$
| u_l v_{w_{\rm s}} + u_{w_{\rm s}}  v_l |^{\beta Q_0 N} | u_l v_{w_{\rm s}} - u_{w_{\rm s}}  v_l |^{\beta Q_1 N},
$
and to replace $|u_w|^{\beta Q_0 Q_1 N^2}$ by $|2 u_{w_s} v_{w_s} |^{\beta Q_0 Q_1 N^2}$. With this done, and specialising to 
\begin{equation}\label{sf}
\beta = 2, \quad
Q_1 N = r, \quad Q_0 N = K, \quad R=1/2
\end{equation}
in (\ref{1D}), we calculate using the working of the proof of Proposition \ref{P2.1} that in place of (\ref{13}) is the change of variables
formula 
 \begin{multline}\label{13d}
 |2 u_{w_s} v_{w_s} |^{2Kr}  \prod_{l=1}^N | u_l v_{w_{\rm s}} + u_{w_{\rm s}}  v_l |^{2r}  | u_l v_{w_{\rm s}} - u_{w_{\rm s}}  v_l |^{2K}
  \prod_{1 \le j < k \le N} | u_k v_j - u_j v_k |^2 \, d\Omega_1 \cdots d \Omega_N \\
  =
{|2 w_{\rm s} |^{2 K r} \over  (1 + |w_{\rm s}|^2)^{Kr + (r + K)N} }
\prod_{l=1}^N { | w_{\rm s} +  z_l |^{2r}    | w_{\rm s} -  z_l |^{2K}  \over   (1 + | z_l|^2)^{K+r+N+1} } \prod_{1 \le j < k \le N} |z_k - z_j|^2 \, d\mathbf r_1 \cdots d \mathbf r_N,
  \end{multline}
where $w_{\rm s} > 1$ is the stereographic projection to the complex plane of the point on the sphere with Cayley-Klein parameters $(u_{w_{\rm s}},v_{w_{\rm s}})$.

\subsection{Equilibrium measure for symmetrically placed charges}
Firstly one notes that general considerations can be used to establish the existence and uniqueness of the equilibrium measure in the present setting \cite{LD21}.
From the discussion below Proposition \ref{P1.1} it suffices to consider the situation when the inequality (\ref{1F}) is violated. It is convenient to consider the
droplet in its form after a stereographic projection from the sphere to the plane. We have two formalisms, corresponding to either (\ref{13}) of (\ref{13d}).

In analogy with (\ref{1H}), in the case (\ref{13})  the boundary of the projected droplet,
$\tilde{\Omega}_{\rm d}$ say, can be determined as the solution of the equation
\begin{equation}\label{E1}
- (Q_0 + Q_1 + 1) \log (1 + |z|^2) + Q_1 \log |w - z|^2 + \int_{\tilde{\Omega}_{\rm d}} \mu(z') \log | z - z'|^2 \, d^2z' = C, \quad z \in \tilde{\Omega}_{\rm d},
\end{equation}
where $\mu(z)$ is the normalised density in the droplet and $C$ is independent of $z$.
Application of the operator $\partial_z \partial_{\bar{z}}$, and use of the fact that $\partial_z \partial_{\bar{z}}(- \log | z - z'|^2) = - {\pi \over 2} \delta(z - z')$ 
(this can be viewed as the solution of the Poisson equation in the complex plane), tells us that
\begin{equation}\label{E2}
\mu(z) = {Q_0 + Q_1 + 1 \over \pi} {1 \over (1 + |z|^2)^2}.
\end{equation}
This is in keeping with the density on the sphere inside the droplet being the constant $(Q_0 + Q_1 + 1)/\pi$; the factor $1/(1 + |z|^2)^2$ can be viewed
as the Jacobian relating to mapping from the area differential on the sphere to the plane under stereographic projection.
In the case of (\ref{13d}), replacing (\ref{E1}) is the condition
\begin{multline}\label{E1d}
- (Q_0 + Q_1 + 1) \log (1 + |z|^2) + Q_0 \log |w_{\rm s} + z|^2 + Q_1 \log |w_{\rm s} - z|^2 \\
+ \int_{\tilde{\Omega}_{\rm d_s}} \mu(z') \log | z - z'|^2 \, d^2z' = C_{\rm s}, \quad z \in \tilde{\Omega}_{\rm d_s},
\end{multline}
with the reasoning leading to (\ref{E2}) again applying, but with the droplet  $\Omega_{\rm d_s}$ distinct from that in (\ref{E1}).

As previously mentioned, the formalism \eqref{E1d}, with symmetrically located charges, offers certain advantages in computations. On the other hand, a particular advantage of the formalism \eqref{E1} is that it allows one to recover \eqref{1Ha} by taking a properly scaled limit. More precisely, let 
$$
W(z;w):=  (Q_0 + Q_1 + 1) \log (1 + |z|^2) - Q_1 \log |w - z|^2
$$
denote the expression (up to sign) appearing on the LHS of \eqref{E1}. Then, one can observe that 
\begin{equation}\label{US}
\lim_{ Q_0 \to \infty } \bigg(   W(z;w) \Big|_{ z \mapsto \frac{z}{ \sqrt{ Q_0} } , w\mapsto \frac{w}{ { Q_0} } } -Q_1 \log Q_0 \bigg) =  |z|^2 - Q_1\log|w-z|^2,
\end{equation}
where the RHS corresponds to twice of the expression appearing on the LHS of \eqref{1Ha} (with $Q_1 \to Q$, $w \to a$ and up to sign). Since the additive constant in such an external potential does not affect the definition of the Coulomb system, this particular scaled limit should recover the droplet described in the previous literature in terms of \eqref{1Hb}. That this is indeed the case is verified at the end of \S \ref{S2.3u}. 

\subsection*{Symmetrically placed charges with $Q_0 = Q_1$}
The main result of \cite{CK19} is that in the case $Q_0 = Q_1$ the droplet in the case that the inequality (\ref{1F}) is violated is a particular ellipse. (It is pointed out in \cite{CK22} that this same result can be deduced from \cite[Example 3]{GT11}.)
The condition for (\ref{1F}) to be violated is \cite{CK19}
\begin{equation}\label{W1}
Q_0 > {1 \over w_{\rm s}^2 - 1}, \quad (w_{\rm s} > 1),
\end{equation}
as can be checked by setting $Q_0 = Q_1$, $\cos^2 {\theta_w \over 2} = \sin^2 \theta_{w_{\rm s}} = (2 w_{\rm s}/(1 + w_{\rm s}^2))^2$.
The derivation does not proceed directly from (\ref{E1d}) but rather via the introduction of a certain dual weighted
energy problem. It is instructive to show how the result of \cite[Th.~1.1 with the role of $x$ and $y$ interchanged]{CK19}  can be obtained directly from (\ref{E1d}).

\begin{prop}\label{P2.2x}
Let the parameters $Q_0, w_{\rm s}$ be as in (\ref{W1}). Consider the conformal mapping
\begin{equation}\label{E2f}
\zeta_{\rm s}(u) = a_1 u + {a_2 \over u}, \quad a_1 = {1 \over 2} ({c_2} - {c_1} ), \: \: a_2 = {1 \over 2} ( {c_2} + {c_1}),
\end{equation}
where
\begin{equation}\label{E2f+}
c_1^2 = {w_{\rm s}^2 + 1 \over 2 (w_{\rm s}^2 Q_0 - Q_0 - 1)}, \quad c_2^2 = {w_{\rm s}^2  - 1 \over 2 (w_{\rm s}^2 Q_0 + Q_0 +1)},
\end{equation}
mapping the interior of the unit disk to the exterior of the ellipse 
\begin{equation}\label{E2g}
{x^2 \over c_2^2} + {y^2 \over c_1^2}  = 1.
\end{equation}
With $Q_0 = Q_1$, $\tilde{\Omega}_{\rm d_s}$ in (\ref{E1d}) is given by the interior of this ellipse.
\end{prop}

Our method of proof requires first establishing an intermediate result.

\begin{prop}\label{P2.2y}
In the case $Q_0 = Q_1$, for $z \in \tilde{\Omega}_{\rm d_s}$ we have
\begin{equation}\label{E4g}
Q_0 \bigg ( {1 \over z +w_{\rm s}} + {1 \over z - w_{\rm s}} \bigg ) + {2 Q_0  + 1 \over 2 \pi i} \int_{\partial \tilde{\Omega}_{\rm d_s}} {1 \over z - u} {\bar{u} \over 1 + u \bar{u}} \, du =0.
\end{equation}
Consequently, there exists a function $H_{\rm s}(u)$ analytic  in $\mathbb C \backslash \tilde{\Omega}_{d_{\rm s}}$ and on the boundary $\partial \tilde{\Omega}_{d_{\rm s}}$,
with the property that $H_{\rm s}(u) \to 0$ as $u \to \infty$,
and such that $w_{\rm s}, -w_{\rm s} \notin  \tilde{\Omega}_{\rm d_s}$
\begin{equation}\label{E5g+}
H_{\rm s}(u) = {\bar{u} \over 1 + u \bar{u}} - {Q_0 \over 2 Q_0  + 1} \bigg ( {1 \over u - w_{\rm s}} + {1 \over u + w_{\rm s}} \bigg ),
 \quad u \in \partial \tilde{\Omega}_{\rm d_s}.
\end{equation}
\end{prop}

\begin{proof}
Acting on (\ref{E1d}) in the case $Q_0 = Q_1$ with $\partial_z$ gives
\begin{equation}\label{E3g+}
- {\bar{z} \over 1 + |z|^2} + {Q_0 \over 2 Q_0  + 1} \bigg ( {1 \over w_{\rm s} - z} + {1 \over w_{\rm s} + z} \bigg )
 + {1 \over \pi} \int_{\tilde{\Omega}_{\rm d_s}} {1 \over z - z'} \,  {d^2 z' \over (1 + |z'|^2)^2}
= 0, \quad z \in \tilde{\Omega}_{\rm d_s}.
\end{equation}
Inside the integral, with $z'=u$ for convenience, we observe
\begin{equation}\label{2.11a+}
{1 \over (1 + |u|^2)^2} = \partial_{\bar{u}} \Big (  {\bar{u} \over 1 + u \bar{u}}  \Big ).
\end{equation}
Use of the Cauchy-Pompeiu formula then gives
\begin{equation}\label{2.19}
{1 \over \pi} \int_{\tilde{\Omega}_{\rm d_s}} {1 \over z - u}  \partial_{\bar{u}} \Big (  {\bar{u} \over 1 + u \bar{u}}  \Big ) \, d^2 u =
{\bar{z} \over 1 +  z \bar{z}} -{1 \over 2 \pi i}  \int_{\partial \tilde{\Omega}_{\rm d_s}} {1 \over u - z} {\bar{u} \over 1 + u \bar{u}} \, du.
\end{equation}
This can be substituted for the final term on the LHS of (\ref{E3g+}) to obtain (\ref{E4g}).

In relation to
the statement relating to the function $H(u)$, one begins by noting from Cauchy's theorem that for $w_{\rm s}, -w_{\rm s} \notin  \tilde{\Omega}_{d_{\rm s}}$
and $z \in \tilde{\Omega}_{d_{\rm s}}$,
\begin{equation}\label{2.18+}
Q_0 \bigg ( {1 \over z +w_{\rm s}} + {1 \over z - w_{\rm s}} \bigg )  = {Q_0  \over 2 \pi i}  \int_{\partial \tilde{\Omega}_{\rm d_s} } {1 \over u - z}  
\bigg ( {1 \over u +w_{\rm s}} + {1 \over u - w_{\rm s}} \bigg )  \, du.
\end{equation}
Now let $H_{\rm s}(u)$ defined on $\partial \tilde{\Omega}_{\rm d_s}$ by (\ref{E5g+}).
Substituting (\ref{2.18+}) in  (\ref{E4g}) then implies,
$$
0 =  {1 \over 2 \pi i}  \int_{\partial \tilde{\Omega}_{\rm d_s}} {H_{\rm s}(u) \over u - z}    \, du.
$$
Reasoning given in \cite{LD21}, which can be traced back to the pioneering work of Richardson
\cite{Ri72},  allows it to be concluded from this that $H_{\rm s}(u)$ can be extended to an analytic
function outside of $\partial \tilde{\Omega}_{\rm d_s}$, decaying at infinity, while maintaining its stated boundary value.
\end{proof}

\begin{remark}
In the case $z$ outside of the droplet, application of the Cauchy-Pompeiu formula gives
$$
{1 \over \pi} \int_{\tilde{\Omega}_{\rm d_s}} {1 \over z - u}  \partial_{\bar{u}} \Big (  {\bar{u} \over 1 + u \bar{u}}  \Big ) \, d^2 u =
 -{1 \over 2 \pi i}  \int_{\partial \tilde{\Omega}_{\rm d_s}} {1 \over u - z} {\bar{u} \over 1 + u \bar{u}} \, du;
$$
cf.~(\ref{2.19}). Now substituting for ${\bar{u} \over 1 + u \bar{u}}$ in the integrand on the RHS according to
(\ref{E5g+}), then applying Cauchy's theorem from the viewpoint of the exterior of the droplet shows
\begin{equation}\label{2.20a}
{1 \over \pi} \int_{\tilde{\Omega}_{\rm d_s}} {1 \over z - z'} \,  {d^2 z' \over (1 + |z'|^2)^2} = H_{\rm s}(z), \quad z \notin \tilde{\Omega}_{\rm d_s}.
\end{equation}
In words we thus have that $H_{\rm s}(z)$ for values of $z$ outside of the droplet is equal to the Stieltjes (also called Cauchy) transform of the droplet.
\end{remark}

Let $\zeta_{\rm s}(u)$ be the conformal map from the interior of the unit disk to the exterior of the droplet. Provided this conformal map is real for $u$ real, we know from \cite[Lemma 2]{EM07} that
 \begin{equation}\label{B2x}
 H_{\rm s}(\zeta_{\rm s}(u)) = {\zeta_{\rm s}(1/u) \over 1 + \zeta_{\rm s}(u) \zeta_{\rm s}(1/u)} - {Q_0 \over 2 Q_0 + 1} \bigg ({1 \over \zeta_{\rm s}(u) + w_{\rm s}}
 + {1 \over \zeta_{\rm s}(u) - w_{\rm s}} \bigg ),
  \quad |u|<1.
 \end{equation}
 The requirement that (\ref{B2x}) be analytic for $|u|<1$ implies equations for $a_1,a_2$ in (\ref{E2f}), which are relevant to deducing (\ref{E2g}).
 On this point, substituting (\ref{E2f}) in (\ref{B2x}) shows that the first term on the RHS has poles at the zeros of
 \begin{equation}\label{4}
 u^4 + {(1 + a_1^2 + a_2^2) \over a_1 a_2} u^2 + 1 = 0.
  \end{equation}
  In the variable $u^2$ it is straightforward to check that the solutions of (\ref{4}) are both real and reciprocals of each other. Furthermore, provided $a_1,a_2$ have
  opposite signs they are both positive. Assuming this, the four solutions of (\ref{4}) can be written $u_0$, $1/u_0$, $-u_0$, $-1/u_0$ for a particular
  $0 < u_0 < 1$.

   For the RHS of (\ref{B2x}) to be analytic for $|u|  < 1$ requires that the poles at $u= \pm u_0$  cancel  the poles of the second term. 
   In relation to their location, this implies
 \begin{equation}\label{4v}   
  \zeta_{\rm s}(u_0) = w_{\rm s}, \qquad    \zeta_{\rm s}(-u_0) = -w_{\rm s}.
  \end{equation} 
  An immediate consequence of this, following too from the
  fact that the first term has a pole at $u_0$, is  that $ \zeta_{\rm s}\Big (  {1 \over u_0} \Big ) =  -{1 \over w_{\rm s}}$. From this equation,
  and the first equation in (\ref{4v}),
 we can use the functional form of $\zeta_{\rm s}(u)$ in (\ref{E2f}) to 
 solve for $u_0$ to obtain
 \begin{equation}\label{4u} 
 u_0 = { w_{\rm s} a_1 + {a_2 \over w_{\rm s}} \over a_1^2 - a_2^2}.
 \end{equation} 
 Substituting back in the first equation of (\ref{4v}) implies
  \begin{equation}\label{4u+} 
  (a_1^2 - a_2^2)^2 + a_1^2 + a_2^2 + a_1 a_2 \Big ( {1 \over w_{\rm s}^2} + w_{\rm s}^2 \Big ) = 0.
 \end{equation}   
 Introducing the variable
 \begin{equation}\label{cc}
 x=-{a_1 \over a_2}, \qquad a_2 > - a_1 > 0
 \end{equation} 
 (we know that $a_1$ and $a_2$ have opposite signs from two lines below (\ref{4}); the assumption on their order with respect to magnitude gives $0 < x < 1$, which recalling $w_{\rm s} > 1$ is a necessary condition for $0<u_0<1$ in (\ref{4u})), this can be solved for $a_2^2$ to give
  \begin{equation}\label{cc1} 
  a_2^2 = {x(1 + w_{\rm s}^4) - (1 + x^2)w_{\rm s}^2 \over (1 - x^2)^2 w_{\rm s}^2}.
   \end{equation} 
   For $a_2^2>0$, again recalling $w_{\rm s} > 1$, we have the further constraint ${1 \over w_{\rm s}^2} < x$.

 We require too that the  residues of the poles in the two terms cancel, which as we will see implies a further equation for $a_1,a_2$, or more specifically for $a_2^2,x$.

  \begin{prop}
 We have
  \begin{equation}\label{B5aY}  
  {Q_0 \over 1 + 2 Q_0 } =   { u_0^2 \zeta_{\rm s}'(u_0) \over u_0^2 \zeta_{\rm s}'(u_0) +  w_{\rm s}^2  \zeta'(1/u_0)},
  \end{equation}
  which implies
  \begin{equation}\label{B5aZ1}
  Q_0 w_s^2 \Big ( a_2^2(x^4 - 1) + x w_s^2 - 1 \Big ) =
  x (1 + Q_0) \Big ( w_s^2 x - 1 + 2 a_2^2 (x^2 - 1) \Big ).
  \end{equation}
  \end{prop}

  \begin{proof}
  For (\ref{B2x}) to be analytic at $u = u_0$ we require that
    multiplying by $\zeta_{\rm s}(u) - w_{\rm s}$ and taking the limit $u \to u_0$ gives 0, which implies
  \begin{equation}\label{B5g}    
  {Q_0 \over 1 + 2 Q_0 } = \lim_{u \to u_0}
  {\zeta_{\rm s}(1/u) ( \zeta_{\rm s}(u) - w_{\rm s}) \over 1 + \zeta_{\rm s}(u) \zeta_{\rm s}(1/u)}.
   \end{equation}  
   Computing the limit using L'H\^opital's rule, making use too of (\ref{4v}), we deduce (\ref{B5aY}).
   
   In relation to   (\ref{B5aZ}), minor manipulation of (\ref{B5aY}) gives
    \begin{equation}\label{D7}
   Q_0 w_{\rm s}^2 \zeta_{\rm s}'(1/u_0) = (1 + Q_0) u_0^2  \zeta_{\rm s}'(u_0).
    \end{equation}    
   The definition of $\zeta_{\rm s}(u)$ in (\ref{E2f}) gives
     \begin{equation}\label{D7a}
    \zeta_{\rm s}'(1/u_0) = a_1  - a_2 u_0^2, \quad u_0^2   \zeta_{\rm s}'(u_0) = a_1 u_0^2 - a_2.
    \end{equation}     
    Further, the definition of $\zeta_{\rm s}(u)$ in (\ref{E2f}) and the first equation in (\ref{4v}) gives
    that
   \begin{equation}\label{D7b}   
    u_0^2 = {1 \over a_1} (w_{\rm s} u_0 - a_2). 
     \end{equation} 
     In this we substitute (\ref{4u}), with the result substituted in (\ref{D7a}), and these results
     then substituted in (\ref{D7}). Simple manipulation then gives  
\begin{multline}\label{B5aZ}  
    Q_0 w_{\rm s}^2 \Big (   (a_1^2 + a_2^2) (a_1^2 - a_2^2) - a_2 (w_{\rm s}^2 a_1 + a_2) \Big ) \\
    = (1 + Q_0) \Big ( a_1(w_{\rm s}^2 a_1 + a_2) - 2a_2 a_1 (a_1^2 - a_2^2) \Big ).
    \end{multline}
    Introducing the variable $x$ in this according to (\ref{cc}) then gives (\ref{B5aZ1}).
    \end{proof}  

  \smallskip
 \noindent {\it Proof of Proposition \ref{P2.2x}. } \: \: Substituting (\ref{cc1}) in (\ref{B5aZ1}) gives a fourth order polynomial in $x$. Moreover, this polynomial has the factors $x$, $(x-1/w_{\rm s}^2)$ and
 \begin{equation}\label{cc2}
1 - {2 (1 + Q_0(1 + w_{\rm s}^4)) \over (1 + 2 Q_0) w_{\rm s}^2 }  x +  x^2 = 0.
\end{equation}
The solutions $x=0$ and $x= 1/w_{\rm s}^2$ are rejected  due to the requirement $x > 1/w_{\rm s}^2 $ noted below (\ref{cc1}). Thus $x$ is given by the unique solution of (\ref{cc2}) between 0 and 1. With this knowledge, $a_1$ and $a_2$ are uniquely determined by (\ref{cc}) and (\ref{cc1}), giving their values as implied by (\ref{E2f}) and (\ref{E2f+}).

\hfill $\square$
\begin{remark}\label{R1} ${}$ \\
1.~A constraint on the droplet shape is the normalisation condition
 \begin{equation}\label{4s} 
 \int_{\tilde{\Omega}_{\rm d_s} } \mu(z) \, d^2 z = 1.
  \end{equation}
  Taking into consideration that the normalised density is given by (\ref{E2}) we can make use of 
  (\ref{2.20a}) to conclude that this condition requires that for small $u$,
  \begin{equation}\label{B5e+}    
  H_{\rm s} \Big ( {1 \over u} \Big ) \sim  {u \over 1 + 2 Q_0}. 
   \end{equation}   
This asymptotic behaviour can be verified from (\ref{B2x}). Thus, from 
    the functional form of $\zeta_{\rm s}(u)$ in (\ref{E2f}), 
  we have that $\zeta_{\rm s}(u) \sim a_2/u$ as $u \to 0$ and also $\zeta_{\rm s}(1/u) \sim {a_1 \over u}$ in the same limit. Use of these behaviours in (\ref{B2x}) gives (\ref{B5e+}), independent of the particular values of $a_1, a_2$. \\
  2.~It also follows from (\ref{2.20a}) that higher order moments of the droplet density, i.e.~values of the average with respect to the normalised droplet density of the monomials $\{z^k\}_{k=1,2,\dots}$, can be read off from higher order terms in the expansion (\ref{B5e+}). With the charges symmetrically placed and equal, only even powers of $z$ average to a nonzero value. In particular, in the case of the second moment $\mu_2$ say, this leads to the exact evaluation
\begin{equation}\label{B5e+1} 
\mu_2 = {1 \over w^2} \Big ( w^4 Q_0 - Q_0 - 1 - \sqrt{(w^4 Q_0 - Q_0 - 1)^2 - w^4} \Big ).
 \end{equation}
We refer to \cite[Remark 2.7]{By24} for a similar discussion for the elliptic Ginibre ensembles.

 \end{remark}

\subsection*{Symmetrically placed charges with $Q_0 \ne Q_1$}
Repeating the considerations of the proof of Proposition \ref{P2.2y} shows that for $Q_0 \ne Q_1$ we have that there exists a function $\tilde{H}(u)$ analytic  in $\mathbb C \backslash \tilde{\Omega}_{\rm d_s}$ 
(up to the boundary), with the property that $\tilde{H}(u) \to 0$ as $u \to \infty$,
and such that for $w_{\rm s} \notin  \tilde{\Omega}_{\rm d_s}$
\begin{equation}\label{E5z}
\tilde{H}(u) = {\bar{u} \over 1 + u \bar{u}} - {1 \over Q_0 + Q_1 + 1} \bigg ( {Q_1 \over u - w_{\rm s}} + {Q_0 \over u + w_{\rm s}} \bigg ) , \quad u \in \partial \tilde{\Omega}_{\rm d_s}.
\end{equation}
As a generalisation of (\ref{E2f}) we propose the four parameter conformal mapping from the interior of the unit disk to the exterior of the droplet
\begin{equation}\label{E2fr}
\tilde{\zeta}(u) = {b_3 + b_2 u + b_1 u^2 \over u + b_0},
\end{equation}
where the $b_j$ are real but otherwise remain to be determined. In the case $Q_0 = Q_1$ we must have $b_3 = a_2$,
$b_1 = a_1$ and $b_2 = b_0 =0$, where $a_1,a_2$ are as in
Proposition \ref{P2.2x}. Using (\ref{E2fr}), analogous to (\ref{B2x}), the analytic function $\tilde{H}(u)$ for $u$ outside the droplet is then specified by
 \begin{equation}\label{B2xq}
 \tilde{H}(\tilde{\zeta}(u)) = {\tilde{\zeta}(1/u) \over 1 + \tilde{\zeta}(u) \tilde{\zeta}(1/u)} - {1 \over  Q_0 +Q_1+ 1} \bigg ({Q_0 \over \tilde{\zeta}(u) + w_{\rm s}}
 + {Q_1 \over \tilde{\zeta}(u) - w_{\rm s}} \bigg ),
  \quad |u|<1,
 \end{equation}
as is consistent with \cite[Lemma 2]{EM07}. Requiring that the poles $u_1$ and $u_2$ say of the first term on the RHS for $|u| < 1$ are such that $\tilde{\zeta}(u_1) = w_{\rm s}$ and $\tilde{\zeta}(u_2) = - w_{\rm s}$, and moreover that their residues cancel with poles of the remaining terms on the RHS, can be used to deduce four equations for the $b_i$. On the other hand, it turns out that the number of unknowns in the required conformal mapping can be reduced from four to three by consideration of the setting when the charge of strength $Q_0N $ is at the south pole. Therefore we do not pursue the details of the four equations, but rather return to consideration of (\ref{E1}).

\subsection{Equilibrium measure with the charge $Q_0 N$ at the south pole}\label{S2.3u}
 Beginning with (\ref{E1}), the method of the proof of Proposition \ref{P2.2y} gives
 that there exists a function $H(u)$ analytic  in $\mathbb C \backslash \tilde{\Omega}_{\rm d}$ 
(up to the boundary), with the property that $H(u) \to 0$ as $u \to \infty$,
and such that for $w \notin  \tilde{\Omega}_{\rm d}$
\begin{equation}\label{E5p}
H(u) = {\bar{u} \over 1 + u \bar{u}} - {Q_1 \over Q_0 + Q_1 + 1}  {1 \over u - w}, \quad u \in \partial \tilde{\Omega}_{\rm d}.
\end{equation}
According to \cite[Eq.~(27) with $N=1$]{CC03} and \cite[Th.~5]{EM07} the boundary  of the droplet is given by the
image of the unit circle under the action of the conformal mapping
\begin{equation}\label{B1t}
\zeta(u) = {R \over u} \Big ( {1 - bu \over 1 - au} \Big ),
\end{equation}
 for certain $R >0$ and
\begin{equation}\label{ab}
0 < a < b ,
\end{equation}
 which furthermore bijectively maps from the interior of the unit disk to the exterior of the droplet 
 --- justification of the inequalities (\ref{ab}) follows from consideration of the limit $w \to \infty$; see Remark \ref{Rr}.
 Since $\zeta(u) \to \infty$ as $u \to 0$, for (\ref{B1t}) to be a bijection,
 it must be that $|a| < 1$. For the same reason we require  $|b| < 1$ if and only if  $0 \notin \tilde{\Omega}_d$. Analogous to
 (\ref{B2x}), the analytic function ${H}(u)$ for $u$ outside the droplet is then specified by \cite[Lemma 2]{EM07}
 \begin{equation}\label{B2xq+}
 {H}({\zeta}(u)) = {{\zeta}(1/u) \over 1 + {\zeta}(u) {\zeta}(1/u)} - {Q_1 \over  Q_0 +Q_1+ 1} {1 \over {\zeta}(u) - w}
  \quad |u|<1.
 \end{equation}

 We consider first the consequence of the pole, to be denoted $v_0$ with the requirement that $|v_0| < 1 $, of the first term on the RHS of (\ref{B2xq+}). To cancel the pole of the second term on the RHS this must be such that 
\begin{equation}\label{zz}
\zeta(v_0) = w. 
\end{equation}
For $v_0$ to correspond to a pole of the first term we must have 
\begin{equation}\label{zzE}
{\zeta}(v_0) {\zeta}(1/v_0)=-1. 
\end{equation}
Hence $v_0$ solves the quadratic equation $A u^2 + B u + C = 0$ with
  \begin{equation}\label{B3+}
A = C = - (a + R^2 b), \quad B = (1 + a^2 + R^2 + R^2 b^2);
 \end{equation}
 see \cite[Eq.~(34) with $N=1$, $b-a \mapsto  -b$]{CC03}. The condition 
  \begin{equation}\label{B3a+}
 B \ge 2|A|
 \end{equation}
 is required for $v_0$ to be real.
In fact a simple calculation shows (\ref{B3+}) is guaranteed for $0 < a < b < 1$ (although the condition 
$b < 1$ is not necessary), and that then
\begin{equation}\label{abY}
a < v_0 < b.
\end{equation}

The equation (\ref{zz}), and its companion ${\zeta}(1/v_0)=-1/w$ implied by (\ref{zzE}), allows for the analogue of (\ref{4u}) to be deduced. Thus, after scaling $a,b$ according to
  \begin{equation}\label{abs}
  a = R \alpha, \quad b = \beta/R
\end{equation}
and eliminating $v_0^2$ in these two equations we obtain
 \begin{equation}\label{abs1}
 v_0 = R \mathcal A, \qquad \mathcal A := \bigg ( {1 + \alpha^2 \over \beta + w + (-\beta w + 1) \alpha } \bigg ).
 \end{equation}
Substituting this back in (\ref{zz}) gives that in terms of the variables (\ref{abs}),
\begin{equation}\label{abs2}
R^2 = {(\beta + w) \mathcal A - 1 \over \alpha w \mathcal A^2}.
\end{equation}

  After the introduction of (\ref{abs}),  the parameters $R,\alpha,\beta$ specify the conformal mapping (\ref{B1t}). The equation (\ref{abs2}) expresses $R^2$ in terms of $\alpha, \beta$. It remains then to determine two more equations linking the unknowns. We do this by considering the normalisation condition (recall Remark \ref{R1}) as well as the residue matching condition in (\ref{B2xq+}).

 \begin{prop}\label{P2.5r}
 Let $v_0$ be characterised by (\ref{zz}). We have
  \begin{equation}\label{B5a}  
  {Q_0 \over 1 + Q_0 + Q_1} = {1 \over 1 + {\beta \over \alpha}}, \quad
  {Q_1 \over 1 + Q_0 + Q_1} =   {v_0^2 \zeta'(v_0) \over v_0^2 \zeta'(v_0) +  w^2  \zeta'(1/v_0)}.
  \end{equation} 
  \end{prop}  
  
  \begin{proof}

Replacing (\ref{B5e+}) as a consequence of the normalisation condition is the requirement that
 \begin{equation}\label{B5d}  
\mathop{\rm Res} \limits_{u = 0} {H(1/u)  \over u^2} = {1   \over 1 + Q_0 + Q_1}.
 \end{equation}  
 On the other hand, 
  from (\ref{B1t}) we have that $\zeta(u) \sim R/u$ as $u \to 0$ and also $\zeta(1/u) \sim {R b \over a} u$ in the same limit.
  Substituting in (\ref{B2xq+}) shows
  \begin{equation}\label{B5f}      
  \mathop{\rm Res} \limits_{u = 0} {H(1/u)  \over u^2} = 
  {{b \over a} R^2  \over 1 + {b \over a} R^2} - {Q_1 \over 1 + Q_0 + Q_1 }. 
    \end{equation}   
    Equating this with (\ref{B5d}), after some simple manipulation and use of the scaled variables (\ref{abs}) we arrive at the first relation in (\ref{B5a}).
    
    The second relation in (\ref{B5a}) is derived in an identical manner to (\ref{B5aY}).
    
\end{proof}  

According to the first equation in Proposition \ref{P2.5r} we have that
 \begin{equation}\label{abd}
 \beta =  {1 + Q_1 \over Q_0}  \alpha.
\end{equation} 
Considering in addition (\ref{abs2}) it remains to specify $\alpha$. This can be done by making explicit the second equation in Proposition \ref{P2.5r}, and substituting for $R$ and $\alpha$. As a first step, by first computing the logarithmic derivative of (\ref{B1t}), making use of (\ref{zz}) and its companion noted above (\ref{abs}),  and use of the scaled variables (\ref{abs}) together with (\ref{abs1}) and (\ref{abs2}), the variable $R$ can be eliminated in this equation to give
 \begin{equation}\label{abd1}
 (1 + Q_0) \bigg ( {1 \over 1 - \alpha \mathcal B} - {1 \over 1 - \beta \mathcal A} \bigg ) - Q_1  \bigg ( {1 \over -1 +  \mathcal B/ \beta} - {1 \over -1 +  \mathcal A/\alpha} \bigg )  = 0.
\end{equation} 
Here $\mathcal A$ is given by (\ref{abs1}) and
 \begin{equation}\label{abd2}
 \mathcal B = {1 \over \alpha w} \bigg ( \beta + w - {1 \over \mathcal A} \bigg ).
\end{equation} 
In (\ref{abd1}) we can further eliminate $\beta$ using (\ref{abd}), with $\alpha$ then the sole unknown. Now using computer algebra to reduce  to rational function form reveals that $\alpha$ satisfies the fourth order polynomial equation
\begin{multline}\label{2.62}
    w Q_0 + \Big ( 1 + 2 Q_0 - (Q_0 + Q_1) w^2 \Big ) \alpha -3 (1 + Q_0 + Q_1) w \alpha^2 \\
    + (-1-2Q_1 + (2 + Q_0 + Q_1 ) w^2) \alpha^3 + (1 + Q_1) w \alpha^4 = 0.
\end{multline}

Significant  is that for $Q_0 = Q_1$ (\ref{2.62}) factorises to give 
 \begin{equation}\label{abd3}
 \Big ( -1 + 2 w  \alpha + \alpha^2 \Big ) \Big ( -Q_0 w -(1 + 2 Q_0)  \alpha  + (1 + Q_0) w \alpha^2 \Big ) = 0.
\end{equation}  
Each factor has one positive and one negative solution, with only the positive solutions of relevance (recall (\ref{ab}) and (\ref{abs})). Of these we reject that corresponding to the second factor, due to the denominator in the quantity $\mathcal A$ as defined in (\ref{abs1}) then vanishing. Hence
 \begin{equation}\label{abd4}
 \alpha = -w + (w^2 + 1)^{1/2},
\end{equation}  
and thus we have explicit knowledge of  the values of all the parameters in the conformal mapping (\ref{B1t}).

\begin{prop}\label{P2.6}
    Consider the setting that the charge $Q_0N$ is at the south pole, and let the position of the charge of strength $Q_1N $ have azimuthal angle $\phi=0$ and polar angle $\theta_w$, which furthermore maps to the point $w > 0$ on the real axis under stereographic projection. Restrict attention to the case $Q_0 = Q_1$, and require that $Q_0, w$ are such that the inequality
    (\ref{1F}) is violated (meaning that the associated spherical caps do intersect), which corresponds to the condition
    \begin{equation}\label{wdc}
        w^2 > {1 \over 4 Q_0 (1 + Q_0)}.
    \end{equation}
    In this setting the boundary of the droplet is given by (\ref{B1t}), where 
    according to (\ref{abs}), (\ref{abs2}), (\ref{abd}) the parameters $R,a,b$ therein are given in terms of $\alpha$, which is specified by (\ref{abd4}). 
    \end{prop}

\begin{figure}
  \begin{minipage}[b]{0.19\linewidth}
    \centering
    \includegraphics[width=\textwidth]{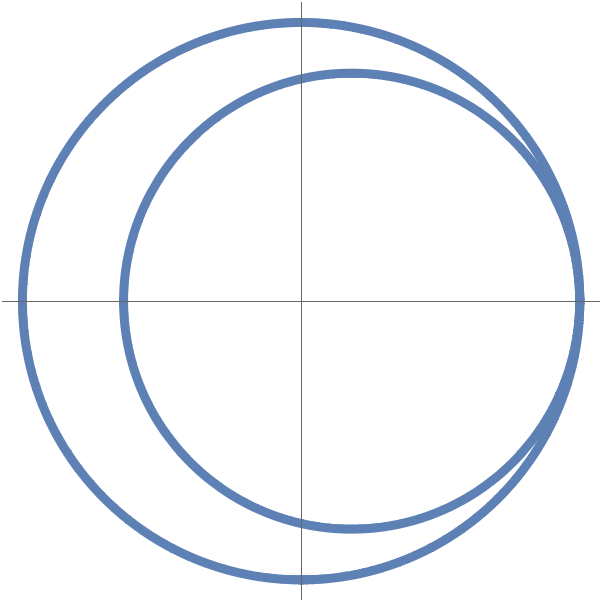}
    \subcaption{$w=w_{ \rm cri } \approx 0.11$}
  \end{minipage}  
  \begin{minipage}[b]{0.19\linewidth}
    \centering
   \includegraphics[width=\textwidth]{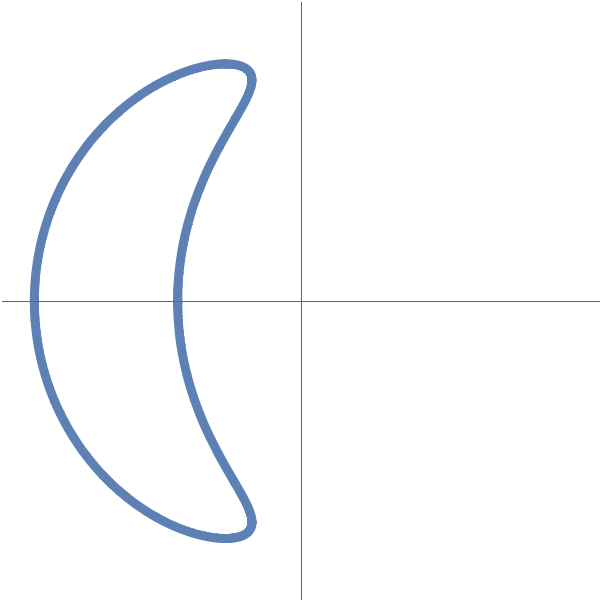}
    \subcaption{$w=0.3$}
  \end{minipage} 
  \begin{minipage}[b]{0.19\linewidth}
    \centering
     \includegraphics[width=\textwidth]{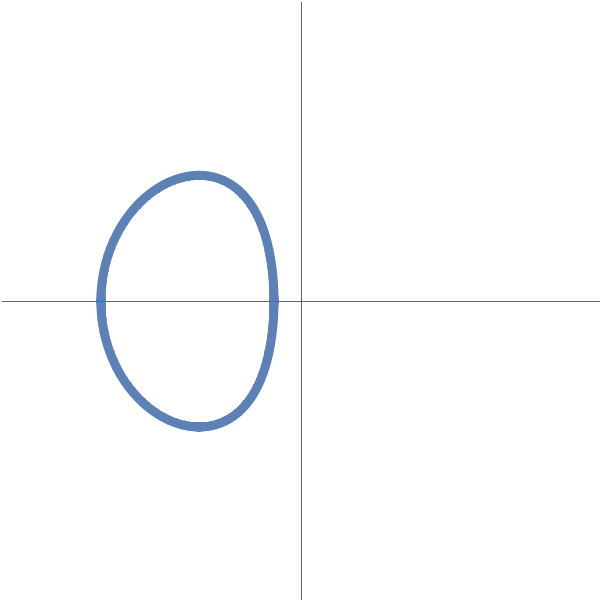}
    \subcaption{$w=1$}
  \end{minipage} 
  \begin{minipage}[b]{0.19\linewidth}
    \centering
 \includegraphics[width=\textwidth]{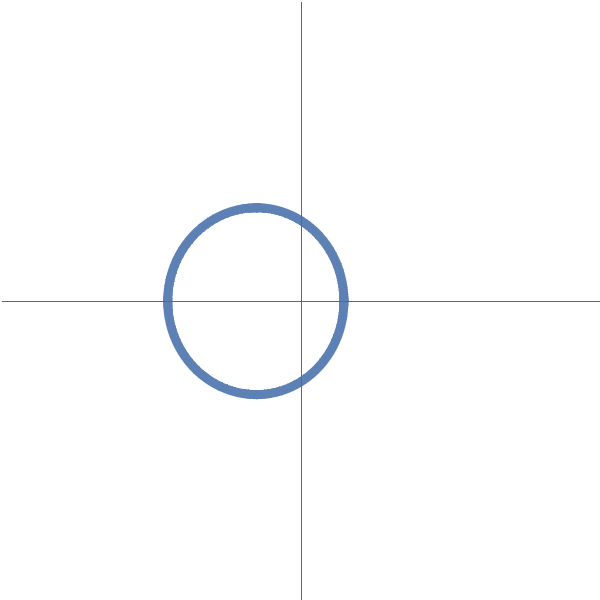}
    \subcaption{$w=3$}
  \end{minipage} 
  \begin{minipage}[b]{0.19\linewidth}
    \centering
   \includegraphics[width=\textwidth]{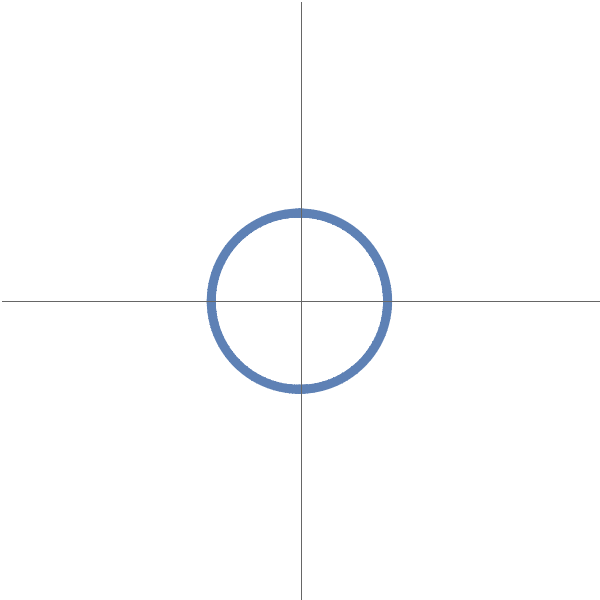}
    \subcaption{$w=80$}
  \end{minipage} 

  \begin{minipage}[b]{0.19\linewidth}
    \centering
    \includegraphics[width=\textwidth]{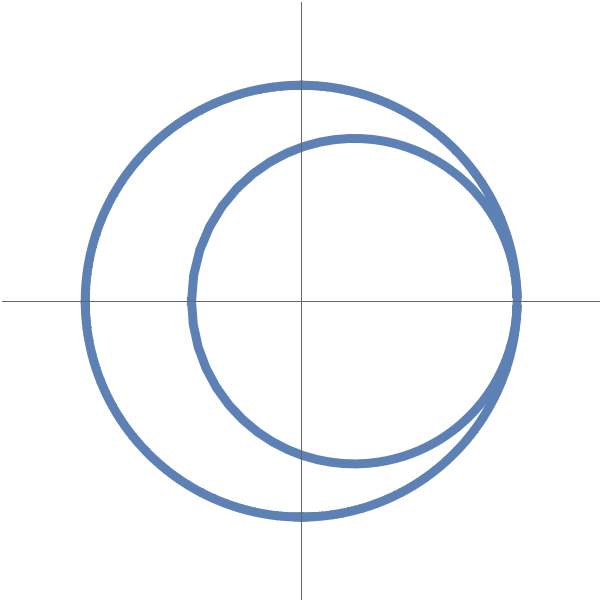}
      \subcaption{$w=w_{ \rm cri } \approx 0.15$}
  \end{minipage}  
  \begin{minipage}[b]{0.19\linewidth}
    \centering
    \includegraphics[width=\textwidth]{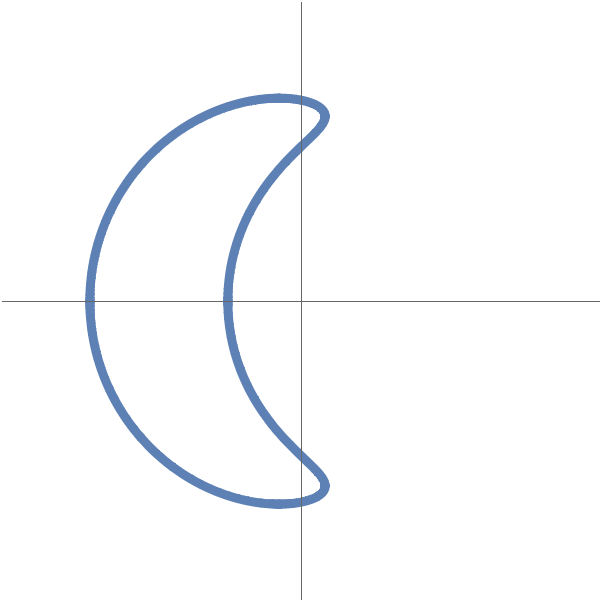}
    \subcaption{$w=0.3$}
  \end{minipage} 
  \begin{minipage}[b]{0.19\linewidth}
    \centering
    \includegraphics[width=\textwidth]{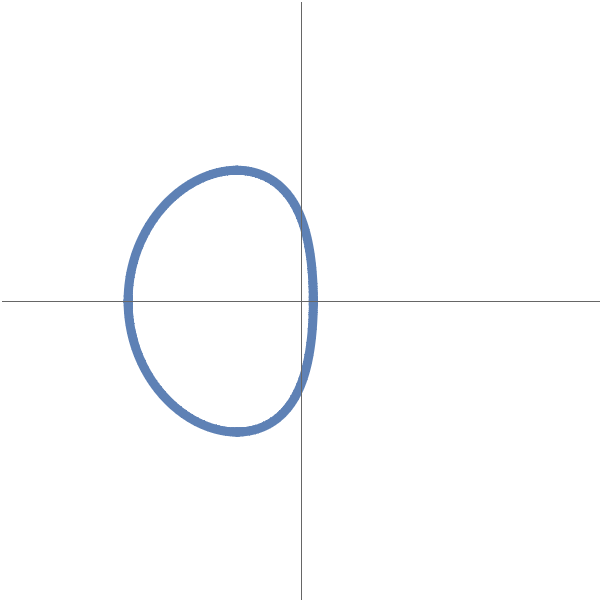}
    \subcaption{$w=1$}
  \end{minipage} 
  \begin{minipage}[b]{0.19\linewidth}
    \centering
    \includegraphics[width=\textwidth]{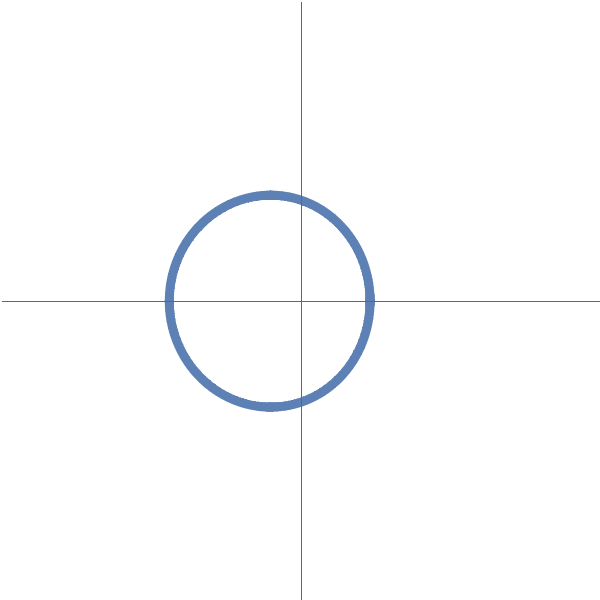}
    \subcaption{$w=3$}
  \end{minipage} 
  \begin{minipage}[b]{0.19\linewidth}
    \centering
    \includegraphics[width=\textwidth]{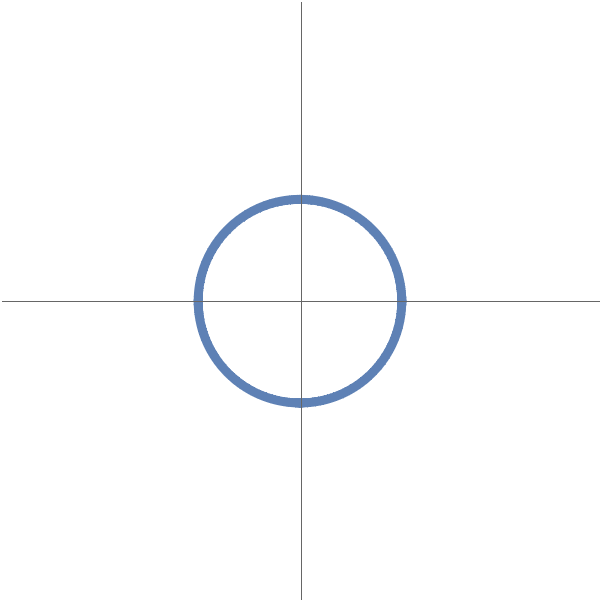}
    \subcaption{$w=80$}
  \end{minipage} 
    \caption{ Plots (a)--(e) show the droplet boundary for the case $Q_0=Q_1=4$, as determined by Proposition \ref{P2.6}, for different values of $w$. Plots (f)--(j) present the corresponding figures for $Q_0=4$ and $Q_1=2$, as determined by Proposition~\ref{P2.7}. In both cases, as $w$ increases, the droplet tends to be a centred disk with a radius of $1/\sqrt{Q_0+Q_1}$, as given in \eqref{Rc+}. Comparing the first and second rows, one can observe that the effect of the point charge at $w$ on the droplet boundary is ``smaller'' in the second.} 
    \label{Fig1}
\end{figure}

Using this result, for a particular $Q_0$ the droplet boundary can be computed and displayed using computer algebra software for varying $w$ obeying (\ref{wdc}); see the first row of Figure \ref{Fig1} for  examples. Some general properties can be identified. First, since from (\ref{abd4}), $\alpha \sim {1 \over 2 w}$ as $ w \to \infty$, we can read off the corresponding leading large $w$ form of $\beta$ from (\ref{abd}), then substitute in (\ref{abs1}) to deduce the first two terms in the large $w$ expansion of $\mathcal A$. Using this information in (\ref{abs2}) then gives
\begin{equation}\label{Rc}
\lim_{w \to \infty} R = \sqrt{1 \over 2 Q_0};
\end{equation}
see \eqref{Rc+} below  for the general $Q_0 \ne Q_1$ case.
This is consistent with (\ref{Y2}) in Appendix A (set $Q_1=0$ and $Q_0 \mapsto 2 Q_0$ to correspond to there being an external charge only at the south pole, with strength $2Q_0N$). The value (\ref{Rc}) can be checked to be in agreement with that seen for the radius of the circle in the final display of the first row of Figure \ref{Fig1}. A circle is also seen as the outer boundary in the limit that the inequality (\ref{wdc}) becomes an equality from above, with the approach to this limit illustrated by the first display of Figure \ref{Fig1}. At the limit the spherical caps no longer overlap, but rather only intersect at a point. This tells us that the value $v_0 = 1$ in (\ref{zz}) must be approached, leading to a breakdown of the requirement of the simultaneous equations $\zeta(v_0) = w$, $\zeta(1/v_0) = -1/w$. Thus in (\ref{abs2}) we must have $(R \mathcal A)^2 = 1$ which in turn implies $\alpha \beta = 1$, or equivalently from (\ref{abs}) that $ab = 1$. Substituting the latter in (\ref{B1t}) shows that then $|\zeta(u)| = R/a = 1/\alpha$ for all $u$ on the unit circle. Moreover, the value of $\alpha$ is specified in terms of $w$ by (\ref{abd4}), then in terms of $Q_0$ by making (\ref{wdc}) an equality; consistency with the numerics can be verified.

The result of Proposition \ref{P2.6} can be compared against a suitable rotation of the result of Proposition \ref{P2.2x} for the droplet boundary in the case of equal strength, symmetrically placed charges about the south pole. The required rotation corresponds to the conformal map
\begin{equation}\label{Rc1}
\eta(z) = {z - 1/w_{\rm s} \over 1 + z/w_{\rm s}},
\end{equation} 
characterised by its mapping in the stereographic projected plane of $-w_{\rm s}$ to infinity, and its antipode $1/w_{\rm s}$ (when viewed on the sphere) to 0. As a direct consequence of Proposition \ref{P2.2x} we therefore have the following alternative description of the droplet boundary in the case of the charge $Q_0N$ at the south pole, and an equal strength charge at $w$.

\begin{cor}\label{C2.1}
    Consider the setting of Proposition \ref{P2.6}. With $\zeta_{\rm s}(u)$ defined as in Proposition \ref{P2.2x}, the boundary of the droplet is given by the image of the unit circle under the conformal mapping
 $\eta(\zeta_{\rm s}(z))$.  
    \end{cor}

See Figure~\ref{Fig_conformal images} for an illustration.

\begin{figure}[t]
\begin{center}
\begin{tikzpicture}[scale=2, every node/.style={align=center}]
\node at (-2.5, 0) {\includegraphics[width=3cm]{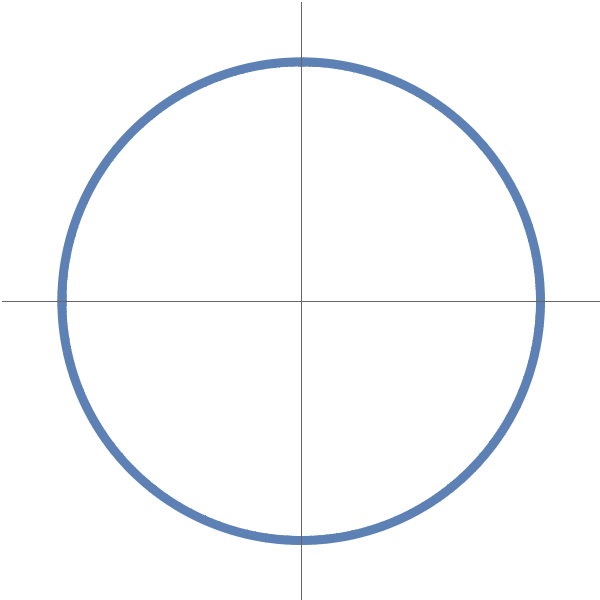}}; 
\draw[->] (-1.6,0) -- (-0.9,0); 
\node[above] at (-1.25, 0) {$\zeta_{ \rm s}(u)$}; 
\node at (0, 0) {\includegraphics[width=3cm]{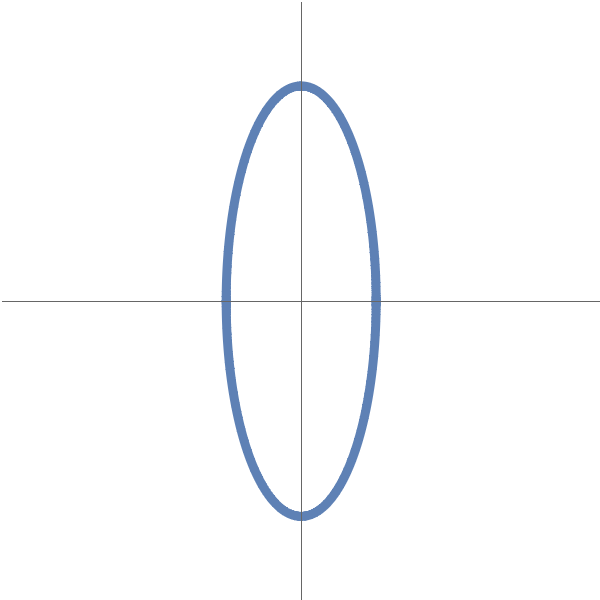}}; 
\node at (2.5, 0) {\includegraphics[width=3cm]{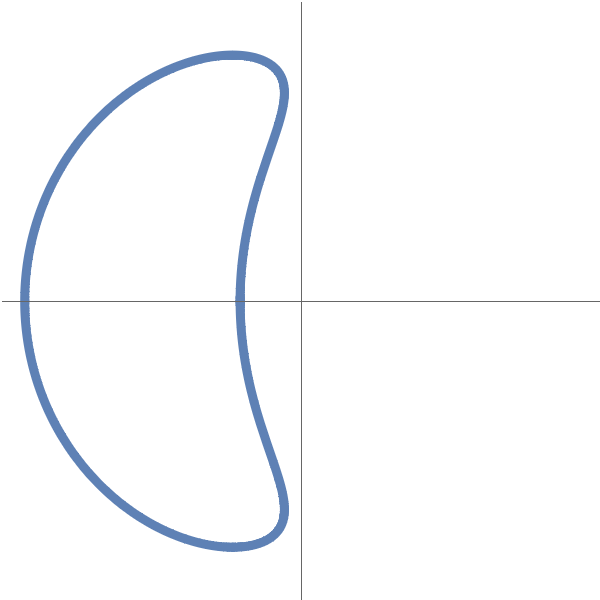}}; 
\draw[->] (0.9,0) -- (1.6,0); \node[above] at (1.25, 0) {$\eta(z)$}; 
\draw[->] (-1.6, 0.5) arc[start angle=110, end angle=70, radius=4.8];
\node[above] at (0, 0.8) {$\zeta(u)$};
\end{tikzpicture}
\end{center}
    \caption{The plot illustrates the conformal mappings $\zeta_{\rm s}$, $\zeta$, and $\eta$ given in \eqref{E2f}, \eqref{B1t}, and \eqref{Rc1}, respectively, as well as the boundary of the droplets $\Omega_{\rm d}$ and $\tilde{\Omega}_{\rm d}$, which are the images of the unit circle under these maps. Here, $Q_0 = Q_1 = 2$ and $w = 2$, so $w_{\rm s} = (1 + \sqrt{5})/2$. }
    \label{Fig_conformal images}
\end{figure}
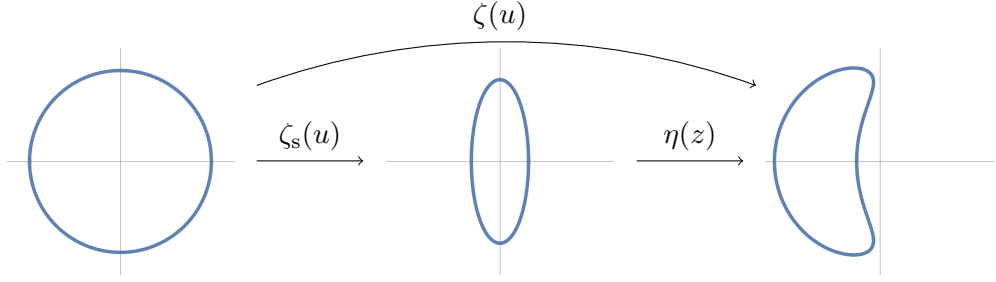

    \begin{remark} ${}$ \\
    1.~While the conformal mappings in Corollary \ref{C2.1} and Proposition \ref{P2.6} have the same image of the unit circle, they are otherwise distinct. \\
    2. The rotation (\ref{Rc1}) must map $w_{\rm s}$ to $w$, so we have
    \begin{equation}\label{Rc2}
    w = {1 \over 2} \bigg ( w_{\rm s} - {1 \over w_{\rm s}} \bigg ), \qquad w_{\rm s} = w + (w^2 + 1)^{1/2},
   \end{equation} 
   with the second equation following as the positive root of the quadratic in $w_{\rm s}$ implied by the first.
   Being able to map between $w$ and $w_{\rm s}$ is necessary for the implementation of Corollary \ref{C2.1}.
   Note too that this allows  for a simplification of (\ref{abd4}), giving
   \begin{equation}\label{Rc3}
   \alpha = {1 \over w_{\rm s}}.
   \end{equation} 
   Substituting in (\ref{abs1}) and (\ref{abs2}) then shows
   \begin{equation}\label{Rc4}
   R^2 = {1 + Q_0(w_{\rm s}^2 + 1) \over 2 Q_0^2 (w_{\rm s}^2 -1 )}, \qquad v_0^2 = {2  w_s^2 \over   Q_0 w_s^4 + w_s^2 - (1 +Q_0)} ;
   \end{equation} 
   note in particular that the first of these equations is consistent with (\ref{Rc}) and the second with the requirement that $v_0 \sim R/w$ for $w$ large as follows from (\ref{B1t}) and (\ref{zz}). The second of these equations is further seen to be consistent with the requirement that $v_0^2 < 1$, once the condition (\ref{wdc}) is accounted for.
     \end{remark}

    We return now to consider the case $Q_0 \ne Q_1$, which is to say the quartic equation (\ref{2.62}).
    An explicit solution analogous to (\ref{abd4}) or
    (\ref{Rc3}) can no longer be expected. By way of analysis, we begin by observing from the signs of the
    constant term, the coefficient of $\alpha^2$, and the coefficient of $\alpha^4$ that there are no more
    than two positive roots. In fact computations carried out by computer algebra of the discriminant give that there are precisely two positive roots as for the case $Q_0 = Q_1$, which furthermore are
    always distinct; see Appendix B. Of these, for consistency with (\ref{abd4}) we must choose the solution which permits the large $w$ expansion
  \begin{equation}\label{2.71} 
  \alpha = \sum_{p=0}^\infty {c_{2p+1} \over w^{2p+1}};
 \end{equation}  
 note that this being an odd function of $w$ is consistent with the parity in $w$ of the coefficients in the powers of $\alpha$ in (\ref{2.62}). Substitution of this form and equating powers of $w$ shows
 \begin{equation}\label{2.72} 
 c_1 = {Q_0 \over Q_0 + Q_1}, \quad c_3 = - {Q_0 Q_1 (Q_0 - Q_1 + Q_0^2 + Q_0 Q_1) \over (Q_0 + Q_1)^4}, \quad \dots
 \end{equation}  
 which with $Q_0 = Q_1$ is indeed consistent with (\ref{abd4}), with the dependence on $Q_0$ then cancelling out entirely. Another characterisation of the solution (\ref{2.71}) is as the smallest positive root, due to
 it approaching zero for $w \to \infty$ and the non-crossing property with the other solution.
 With this established, we have available the specification of all parameters in the conformal mapping (\ref{B1t}), thus allowing for Proposition
 \ref{P2.6}
 to be extended to the general case.

 \begin{prop}\label{P2.7}
 Consider the setting of Proposition
 \ref{P2.6}, but without imposing the restriction $Q_0 = Q_1$. 
 In particular, the 
 requirement that the spherical caps associated with two charges overlap restricts $Q_0,Q_1,w$ to be such that
 \begin{equation}  \label{def of w critical point}
 w  > w_{\rm cri} := \Big( 2Q_0Q_1+Q_0+Q_1 + 2\sqrt{ Q_0 Q_1(1+Q_0)(1+Q_1) } \Big)^{-1/2}    
 \end{equation}
 as follows from (\ref{1F});
 cf.~(\ref{wdc}).
 With this assumed, we have that the boundary of the droplet is given by
 (\ref{B1t}), with the parameters $R,a,b$ given in terms of $\alpha$ 
 according to (\ref{abs}), (\ref{abs2}), (\ref{abd}), while $\alpha$ itself is specified as the solution
 (\ref{2.62}) which exhibits the large $w$ form (\ref{2.71}), (\ref{2.72}) and is the smallest positive root.
 \end{prop}

\begin{remark}\label{Rr} ${}$ \\
 1.~The reasoning relating the case $w \to \infty$ given below Proposition
 \ref{P2.6} carries over, with (\ref{Rc}) now replaced by
 \begin{equation}\label{Rc+}
 \lim_{w \to \infty} R = \sqrt{1 \over Q_0 + Q_1}.
  \end{equation}  
This is in consistent with the fact that the equilibrium measure associated with the rotationally symmetric potential $(Q_0+Q_1+1)\log(1+|z|^2)$ is a centred disk of radius $1/{\sqrt{Q_0+Q_1}}$; see e.g. \cite[\S 5.2]{BF24}.
  Similarly we have prediction that as the negation of (\ref{1F}) begins to break down by the inequality becoming an equality (i.e.~at the pre-critical, post-critical boundary), the support of the droplet will involve an outer circle of radius
  $1/\alpha$. These quantitative properties are borne out by numerical computations, which are displayed qualitatively in the second row of Figure \ref{Fig1}. \\
  2.~Using the large $w$ forms (\ref{2.71}) ($p=0$ term only)
and (\ref{Rc+}), as well as the relation (\ref{abd}) between $\alpha$ and $\beta$, it follows from
(\ref{abs}) that for large $w$
\begin{equation}\label{Rd+}
a \sim{Q_0 \over (Q_0 + Q_1)^{3/2}} {1 \over w}, \qquad {b \over a} \sim (1 + Q_1) \Big (1 + {Q_1 \over Q_0}
\Big ),
 \end{equation}  
 which justifies (\ref{Rr}) by continuity of $a,b$ with respect to $w$, and the fact that the cases
 $a=b$ and $a=0$ can be excluded due the analytic structure of (\ref{B1t}) then being incompatible with
 the cited theory. As part of this remark we recall too the  large $w$ form
 $v_0 \sim R/w$ as previously noted below (\ref{Rc4}).
 \end{remark}
 
  Particular to the case $Q_0 \neq Q_1$ is a scaling limit involving $Q_0 \to \infty$ and $w \to 0$, which in fact relates to the GinUE with an external charge model  of \cite{BBLM15}. To see this we first note that replacing (\ref{E4g}) in the case $Q_0 \ne Q_1$ and with the charge $NQ_0$ at the south pole is the equation
 \begin{equation}\label{S1}
 {Q_1 \over z - w} + {Q_0 + Q_1 + 1 \over 2 \pi i} \int_{\partial \tilde{\Omega}_{\rm d}}
 {1 \over z - u} {\bar{u} \over 1 + u \bar{u}} \, du = 0, \quad z \in \tilde{\Omega}_{\rm d}.
 \end{equation}  
 Suppose now that we scale
  \begin{equation}\label{S1a}
  z \mapsto \epsilon z, \quad w \mapsto \epsilon w, \quad Q_0 = {1 \over \epsilon^2}, \quad 
  \partial \tilde{\Omega}_{\rm d}\mapsto {1 \over \epsilon} \partial \Sigma_{\rm d}.
   \end{equation} 
After changing variables $u \mapsto \epsilon u$ in the integral and taking the limit $\epsilon \to 0$,
   (\ref{S1}) then reads
 \begin{equation}\label{S2}
 {Q \over z - w} + {1 \over 2 \pi i} \int_{\partial {\Sigma}_{\rm d}}
 {\bar{u}  \over z - w}  \, du = 0, \quad z \in {\Sigma}_{\rm d},
 \end{equation} 
 where for convenience we have set $Q_1 = Q$.
 This relation is indeed a consequence of the GinUE with an external charge equilibrium equation (\ref{1Ha}) (after identifying $w$ with $a$, 
  ${\Sigma}_{\rm d}$ with $D^{\rm d}$, and setting $\rho_{\rm b} = {1 \over \pi}$),
 as follows from the working the first paragraph of the proof of Proposition \ref{P2.2y}. Note that this finding is consistent with (\ref{US}).

 For the scaling of $ \partial \tilde{\Omega}_{\rm d}$ to be well defined, we must have the additional scalings
  \begin{equation}\label{S3}
  R \mapsto \epsilon R, \quad \alpha \mapsto {\alpha \over \epsilon}, \quad \beta \mapsto \epsilon \beta,
 \end{equation}  
 as follows from (\ref{B1t}) and (\ref{abs}). Using (\ref{S1a}) and (\ref{S3}) as appropriate in  (\ref{2.62}) implies that in the limit $\epsilon \to 0$ the rescaled $\alpha$
 satisfies the cubic equation
 \begin{equation}\label{S4}
 2 - 3 w \alpha + (-1-2Q + w^2) \alpha^2 + (1 + Q) w \alpha^3 = 0.
  \end{equation} 
  We again require the solution with a large $w$ series form (\ref{2.71}), for which substitution shows
  \begin{equation}\label{S4a} 
c_1 = 1, \quad c_3 = -Q, \quad c_5 = Q (2Q-1), \quad c_7 = Q (-5Q^2 + 7 Q - 1), \dots
  \end{equation}  
  Also, from (\ref{abs}), (\ref{abs2}) and (\ref{abd}), for the rescaled values of $\beta$
 and $R$ we have
  \begin{equation}\label{S5}
  \beta = (1+Q)\alpha, \qquad R^2 = {(1+Q) \alpha^2 + w \alpha - 1 \over w \alpha^3}= {1 \over w \alpha(2 - w \alpha)},
 \end{equation} 
where the second equality in the formula for $R^2$ follows by using (\ref{S4}) in the first equality. Agreement is seen with (\ref{1Hb}) upon identifying $r$ therein with $R$, $q=R \alpha$, $a=w$, together with $\kappa = R \alpha^2 (1 + Q - R^2)$. Thus with these substitutions the second equation in (\ref{1Hb}) can be solved for $R^2$ to give the second equality in the second formula of (\ref{S5}). Next making the same substitution for $q$ in (\ref{1Hc}), then substituting for $R^2$ according to the second equality in second formula of (\ref{S5}) reduces this sixth order polynomial equation in $q$ to the third order polynomial equation in $\alpha$ (\ref{S4}). In relation to $\kappa$, we begin by rewriting the equation to be verified as $q \kappa = R^2 \alpha((1 + Q)\alpha^2 - q^2)$. Now using the first equality in the second equation of (\ref{S5}) reduces the RHS of this to $R^2 \alpha (1 - q^2) (1 - w \alpha)$. Use of the second equality in the second equation of (\ref{S5}) allows this to be identified with the value of $q \kappa$ in (\ref{1Hb}).

\subsection{Electrostatic energies}\label{S2.4r}
Here we take up the question of deriving the replacement of $\mathcal K_N$, $\mathcal K_N^{\rm pre}$ say,
in the Coulomb formulation the computation of the Boltzmann factor \cite[\S 1.4.1]{Fo10}
for the precritical regime. From this viewpoint, and with the surface of the sphere stereographically projected to the plane, in the pre-critical regime there is a background charge density $- N \mu(z)$, with $\mu(z)$ given by (\ref{E2}), restricted to the domain $\tilde{\Omega}_{\rm d}$. In this domain there are $N$ mobile unit charges.

The electrostatic energy due to the interaction of the background with the unit charges
is given by
\begin{align} \label{2.8a}
U_2 & :=  N \sum_{j=1}^N  \int_{\tilde{\Omega}_{\rm d}} \mu(z') \log |z_j - z'| \, d^2 z' \nonumber \\
& = {N \over 2} \bigg ( N C  +  \sum_{j=1}^N \Big ( (Q_0 + Q_1 + 1) \log (1 + |z_j|^2) -
Q_1 \log | w - z_j|^2 \Big ) \bigg ),
\end{align}
where the equality follows from (\ref{E1}). Note that the constant $C$ in (\ref{E1}), which shows itself
again in (\ref{2.8a}), is yet to be determined. Making the explicit choice $z = \zeta(1)$ in (\ref{E1}),
which from the defining property of $\zeta(u)$ relating to $\partial \Omega_{\rm d}$ corresponds to the
right boundary of the droplet on the real axis, we read off the form
\begin{equation} \label{2.8r}
C = - (Q_0 + Q_1 + 1) \log (1 + |\zeta(1)|^2) + Q_1 \log |w - \zeta(1) |^2 + W(\zeta(1)),
\end{equation}
where
\begin{equation}\label{Xx4} 
W(Z):=\int_{\tilde{\Omega}_{\rm d}} \log | Z - z|^2 \mu(z) \, d^2 z.
\end{equation}
The electrostatic energy of the interaction of the background with itself is given by
\begin{align} \label{2.8b}
U_3 & := - {1 \over 2}  \int_{\tilde{\Omega}_{\rm d}} d^2 z \, \mu(z) \int_{\tilde{\Omega}_{\rm d}} d^2 z' \, \mu(z') \log | z - z'|  \nonumber \\
& = - {N^2 \over 4} \bigg ( C   +  (Q_0 + Q_1 + 1) \int_{\tilde{\Omega}_{\rm d}} \log (1 + |z|^2) \mu(z) \,
 d^2 z -  
Q_1  W(w)
 \bigg ),
\end{align}
where (\ref{E1}) has been used to reduce the double integral to a single integral.
From the Coulomb gas construction of the Boltzmann factor \cite[\S 1.4.1]{Fo10}, one has that $\mathcal K_N^{\rm pre}$ is the
factor independent of $\{z_j\}$ in $e^{-\beta(U_2 + U_3)}$ (and also that the terms independent of
$\{z_j\}$ in $U_2 + U_3$ are equal to the  the energy functional (\ref{E1y}) after stereographic projection), and so
\begin{multline} \label{2.8c}
\mathcal K_N^{\rm pre} = \exp \bigg ( - {\beta N^2 \over 4} C + 
{\beta N^2 \over 4} (Q_0 + Q_1 + 1) \int_{\tilde{\Omega}_{\rm d}} \log (1 + |z|^2) \mu(z) \,
 d^2 z 
 - {\beta N^2 Q_1 \over 4} W(w) \\
 +{\beta N^2 Q_1 \over 2} (1 + Q_0) \log (1 + |w|^2) \bigg ).
 \end{multline}
 Here the term on the final line compensates for the factor $1/(1+|w|^2)^{r(K+N)}$ on the RHS of (\ref{13}) when the parameters are specialised as in (\ref{sf}).
Thus, after substituting (\ref{2.8r}), 
\begin{multline}\label{2.8c+}
\frac{4}{\beta N^2} \log \mathcal K_N^{\rm pre}  =  (Q_0 + Q_1 + 1) \log (1 + |\zeta(1)|^2) - Q_1 \log |w - \zeta(1) |^2 - W(\zeta(1))
\\
+(Q_0 + Q_1 + 1) \int_{\tilde{\Omega}_{\rm d}} \log (1 + |z|^2) \mu(z) \,  d^2 z - Q_1 W(w) +
2 Q_1(1+Q_0)  \log (1 + |w|^2). 
\end{multline}

Explicit evaluation of the integral (\ref{Xx4}) specifying $W(z)$, and the one on the second line of (\ref{2.8c+}), is in fact possible. We begin with the integral appearing explicitly in the second line.

 \begin{prop}\label{P8x}
 Let $\zeta(u)$ be specified by (\ref{B1t}), let $a$ be as therein, let $v_0$ be as in (\ref{zzE}),
 and define $c$ as in (\ref{Xx1}) below. We have
  \begin{multline} \label{2.9s} 
  \int_{\tilde{\Omega}_{\rm d}} \log (1 + |z|^2) \mu(z) \, d^2 z 
 =  1 - Q_0 \log \Big ( 1 + {1+Q_1 \over Q_0} \Big ) + {w^2 \zeta'(1/v_0) \over v_0^2 \zeta'(v_0)} Q_1
 \log \Big ( {1  - v_0^2 \over 1  - a v_0} \Big ) \\
 + Q_0 \log \Big ( {v_0 \over a} \Big ) 
 - (Q_0 + Q_1 + 1) \log \Big ( {1  - v_0 b \over 1 - ab} \Big ) - {Q_1} \log  ( |c| (v_0 - a)).
 \end{multline}
 \end{prop}

 \begin{proof}
     Making use of the identity
 \begin{equation} \label{2.8s}   
 {\log (1 + z \bar{z}) \over (1 + z \bar{z})^2} 
 = {1 \over (1 + z \bar{z})^2} - \partial_{\bar{z}}  \Big ( { \log (1 + z \bar{z}) \over z(1 + z \bar{z}) } \Big ) ,
 \end{equation}
 recalling the explicit form of $\mu(z)$ (\ref{E2}), and using the complex form of Green's theorem we have
  \begin{multline} \label{2.8t} 
 \int_{\tilde{\Omega}_{\rm d}} \log (1 + |z|^2) \mu(z) \, d^2 z = 1 - {Q_0 + Q_1 + 1 \over 2 \pi i}
 \int_{\partial \Omega_{\rm d}} { \log (1 + z \bar{z}) \over z(1 + z \bar{z}) } \, dz \\
 = 1 +  {Q_0 + Q_1 + 1 \over 2 \pi i} \oint_{|u|=1}
{ \log (1 + \zeta(u) \zeta(1/u)) \over \zeta(u) (1 + \zeta(u) \zeta(1/u)) } \zeta'(u) \, du,
\end{multline}
where the second equality follows by parameterising $\partial \Omega_{\rm d}$ using $\zeta(u)$. (There is the subtlety that since the conformal mapping is from the inside of the unit circle to the outside of the droplet, the direction of traversing the boundary is reversed, and thus the change of sign.)

Due to the factor $ \zeta'(u)/\zeta(u)$, there is a simple pole at $u=0$. We thus deform the contour from
the unit circle $|u|=1$ to a new closed contour $\mathcal C$ which extends the unit circle by running
along the real axis from $u=-1$ on the upper half plane side to encircle $u=0$ in a small loop, then return to
$u=-1$ along the real axis on the lower half plane side. For this to leave the value of the integral unchanged,
we must add $2 \pi i$ times the residue at $u=0$. This can be computed making use of the first relation in 
(\ref{B5a}) to give in place of the RHS of (\ref{2.8t})
 \begin{equation} \label{2.8n} 
  1 - Q_0 \log \Big ( 1 + {1+Q_1 \over Q_0} \Big ) +  {Q_0 + Q_1 + 1 \over 2 \pi i} \oint_{\mathcal C}
{ \log (1 + \zeta(u) \zeta(1/u)) \over \zeta(u) (1 + \zeta(u) \zeta(1/u)) } \zeta'(u) \, du.
 \end{equation}

Consider $1/(2 \pi i)$ times the contour integral in (\ref{2.8n}), which we will
denote $I_1$. From the property of
$v_0$ (\ref{zzE}) we see that the integrand has a branch point at $v_0$. 
Examination of the functional form of $\zeta(u) \zeta(1/u)$ from (\ref{B1t}) and recalling the inequality (\ref{abY}), we see that inside $|u| = 1$ the integrand is discontinuous along the segment of the real axis
$(a,v_0)$ due to the argument of the logarithm function then being negative. By Cauchy's theorem, we can shrink the contour $\mathcal C$  to begin at $u_0 + \epsilon$ ($0 < \epsilon \ll 1$), traverse the half circle $u_0 + \epsilon e^{i \theta}$, $0 < \theta < \pi$, run along the
interval on the real axis starting at $u_0 - \epsilon$, staying in the upper half plane side of the real axis, until it reaches the point $a$. Then the contour is to change direction, closing by following the mirror image of what was just traversed, now on the lower half plane side of the real axis.

To determine the contribution from the small circle about $v_0$ to $I_1$, we note that for $u \to v_0$,
\begin{equation}\label{Xx1}
{\zeta'(u) \over \zeta(u)} \to {\zeta'(v_0) \over w}, \quad 1+ \zeta(u) \zeta(1/u) \approx c (u - u_0) 
\: \: {\rm with} \: \:  c = - {\zeta'(v_0) \over w} \Big ( {1 + Q_0 + Q_1 \over Q_1} \Big ),
\end{equation}
where the second formula follows from a first order Taylor expansion and making use of the second formula in (\ref{B5a}).
Parametrising this circle as $u = v_0 + \epsilon e^{i \theta}$, $- \pi < \theta < \pi$, and substitution of
(\ref{Xx1}) in the integrand as is valid for $0< \epsilon \ll 1$ shows that the leading contribution to $I_1$ from this portion of the deformed contour equals
\begin{equation}\label{Xx2}
- {Q_1 \over 1 + Q_0 + Q_1} \log(|c| \epsilon).
\end{equation}
For the integral along both sides of the branch cut, only the imaginary part of the logarithm, which is equal to $\pi i$ (resp., $-\pi i$) for when the contour is in the upper (resp., lower) half plane, does not cancel out, with the contribution to $I_1$
then equalling
\begin{equation}\label{Xx3}
- \int_a^{v_0 - \epsilon} { \zeta'(u) \over \zeta(u) (1 + \zeta(u) \zeta(1/u)) } \, du.
\end{equation}
The integral $I_1$ is the sum of (\ref{Xx2}) and (\ref{Xx3}) in the limit $\epsilon \to 0^+$. Substituting  in (\ref{2.8n}) gives, upon minor manipulation to permit the computation of the limit, 
\begin{multline}
1 - Q_0 \log \Big ( 1 + {1+Q_1 \over Q_0} \Big ) -(Q_0 + Q_1 + 1) \\ \times
\int_a^{v_0}  \bigg (  { \zeta'(u) \over \zeta(u) (1 + \zeta(u) \zeta(1/u)) }
  - {Q_1 \over 1 + Q_0 + Q_1} {1 \over v_0 - u} \bigg )
  \, du   - {Q_1} \log  ( |c| (v_0 - a)) .
\end{multline}

The integral  in this expression can be evaluated by noting  that the integrand is a rational function of $u$ with poles at $u={1 \over v_0}$, 0 and ${1 \over b}$, and converting to partial fraction form making use of (\ref{B5a}),
with the result being  (\ref{2.9s}).
\end{proof}

\begin{remark}\label{R2.5} ${}$ \\
1.~In the case $w \to \infty$ we know from Remark \ref{Rr} that the support of the droplet is a disk of radius (\ref{Rc+}). Hence we must have
\begin{multline}\label{2.97}
\lim_{w \to \infty}  \int_{\tilde{\Omega}_{\rm d}}  \log (1 + |z|^2)  \mu(z) \, d^2 z 
 =  {(1+Q_0 + Q_1) \over \pi} \int_{|z|<1/\sqrt{1/(Q_0+Q_1)}} 
 {\log (1 + |z|^2) \over (1 + |z|^2)^2} 
 \, d^2 z  \\
 = 1 - (Q_0 + Q_1) \log \Big ( 1 + {1 \over Q_0 + Q_1} \Big )
.
\end{multline}
To compare against the RHS of (\ref{2.9s}), for $w \to \infty$ we note from Remark \ref{Rr}.2 that the third and fifth terms do not contribute. For the fourth term the limiting value is
$Q_0 \log (1 + Q_1/Q_0)$. In relation to the final term,
from the definition of $c$ in (\ref{Xx1}) one has $c \sim w (1 + Q_0 + Q_1)/(R Q_1)$, allowing us to compute
$$
\lim_{w \to \infty}  \log  ( |c| (v_0 - a)) = \log \Big ( {1 + Q_0 + Q_1 \over Q_0 + Q_1} \Big ).
$$
Substituting these limiting forms in the RHS of (\ref{2.9s}) gives agreement with (\ref{2.97}). \\
2.~Consider the case $Q_0 = Q_1$, and set $Q_0 = \mu /(w_s^2 - 1)$, $\mu > 1$. From (\ref{Rc2}) and 
(\ref{wdc}) one sees that $ \mu \to 1^+$ corresponds to the boundary with the post-critical regime. In the case $Q_0 = Q_1$ we have available too the analytic formulas (\ref{Rc3}) and (\ref{Rc4}). With the aid of computer algebra, we find that this limit can be computed on the RHS of (\ref{2.9s}), giving the simplified functional form
\begin{equation}\label{2.97x} 
1 + \log 2 - (2Q_0 -1) \log \Big ( 1 + {1 \over 2 Q_0} \Big ) - \log \Big (1 + {1 \over Q_0} \Big ).
\end{equation}

\end{remark}

\smallskip
 Recalling (\ref{2.8r}), the remaining integrals in (\ref{2.8c}) are of the form (\ref{Xx4}),
for $Z$ on or outside of the boundary of $\tilde{\Omega}_{\rm d}$. To simplify this class of integral, 
use will be made of the analogue of (\ref{2.20a}) in relation to $\tilde{\Omega}_{\rm d}$ and $H(z)$, which upon substituting
(\ref{B2xq+}) reads 
\begin{equation}\label{Xx5} 
{1 \over Q_0 + Q_1 + 1} \int_{\tilde{\Omega}_{\rm d}} { \mu(z) \over \zeta(u) - z} \, d^2 z =
{\zeta(1/u) \over 1 + \zeta(u) \zeta(1/u)} -
{Q_1 \over Q_0 + Q_1 + 1} {1 \over \zeta(u) - w}.
\end{equation}

\begin{prop}\label{P8y}
Let $\zeta(u)$ be specified by (\ref{B1t}), let $R,a$ be as specified in Proposition \ref{P2.7} and
let $v_0$ be given by (\ref{abs1}). Also, define $v_1$ by (\ref{U1}) below
and let $u_Z$ with $0 < u_Z \le 1$ be such that $Z = \zeta(u_Z)$. With $W(Z)$ defined by (\ref{Xx4}) we have
\begin{equation}\label{Xx91} 
{1 \over 2} W(Z) =  -(Q_0 + 1) \log (1 - u_Z a) + Q_1 {w^2 \zeta'(1/v_0) \over v_0^2 \zeta'(v_0)} \log (1 - u_Z v_0) - Q_1 \log \Big ( 1 - {u_Z \over v_1}  \Big ) +  \log  \Big ({ R \over u_Z} \Big ).
\end{equation}
\end{prop}

\begin{proof}
We multiply both sides of (\ref{Xx5}) by $\zeta'(u)$, integrate from $u_1$ to $u_Z$ and take the real part of both sides to deduce
\begin{multline}\label{Xx6}
{1 \over 2 (Q_0+Q_1+1)} \Big ( W(\zeta(u_Z)) - W(\zeta(u_1)) \Big )  \\ =
\int_{u_1}^{u_Z} { \bigg (
{  \zeta(1/u) \over  (1 + \zeta(u) \zeta(1/u)) }} 
-
{Q_1 \over Q_0 + Q_1 + 1} {1 \over \zeta(u) - w} \bigg )  \zeta'(u) \, du.
\end{multline}
Now take $u_1 \to 0^+$ using that $\zeta(u_1) \sim {R \over u}$ and ${1 \over 2} W(\zeta(u_1)) \sim \log(R/u_1)$. Moving the term involving $W(\zeta(u_1))$ to the RHS then shows that the limit $u_1 \to 0^+$ is
well defined.
As a further manipulation, the fact the integrand (which itself is $H(\zeta(u)) \zeta'(u)$)
has a simple pole at $u=0$ with residue given by the negative of (\ref{B5d}) can be compensated for by
adding (\ref{B5d}) times $1/u$ to the integrand, with the value of this after integration then subtracted.
The value of $u_1$ in the terminals of integration can now be set to zero, giving 
\begin{multline}\label{Xx9} 
{1 \over 2( Q_0 + Q_1 + 1)} W(Z) = 
\int_{0}^{u_Z} { \bigg ( {  \zeta(1/u) \zeta'(u) \over  (1 + \zeta(u) \zeta(1/u)) }} 
- {1\over Q_0 + Q_1 + 1} \Big ( {Q_1  \zeta'(u) \over \zeta(u) - w} - {1 \over u} \Big )\bigg )  \, du \\
+
{ 1 \over Q_0 + Q_1 + 1} \log  \Big ({ R \over u_Z} \Big ). 
\end{multline}

In relation to the integral in this expression, we have from (\ref{B1t}) and (\ref{zzE}) that the integrand is a rational function of $u$ with simple poles at $u=1/a, 1/v_0$, and a further point $u=v_1$, where $v_1 > 1$ has the property $\zeta(v_1) = w$; cf.~(\ref{zz}). From (\ref{B1t}), the equation $\zeta(u) = w$ is quadratic in $u$, so the two roots $v_0, v_1$ (since $v_0<1$ the statement that $v_1>1$, is a consequence of the map $\zeta$ being injective in the unit disc) are simply related,
\begin{equation}\label{U1}
v_1 = {R \over a w} {1 \over v_0}.
\end{equation}
Converting to partial fractions form allows the integrals to be computed, giving (\ref{Xx91}). 
\end{proof} 

\begin{figure}
\begin{minipage}[b]{0.48\linewidth}
    \centering
    \includegraphics[width=\textwidth]{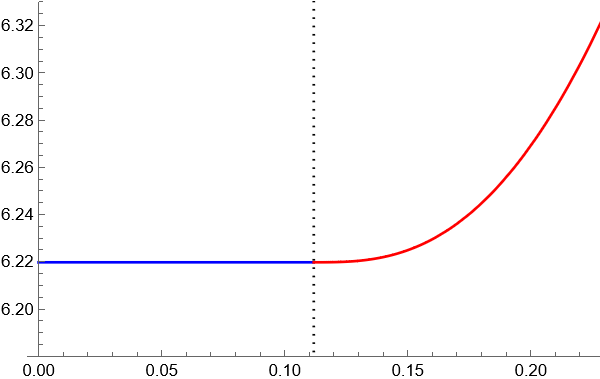}
    \subcaption{$Q_0=Q_1=4$}
  \end{minipage} \quad 
   \begin{minipage}[b]{0.48\linewidth}
    \centering
    \includegraphics[width=\textwidth]{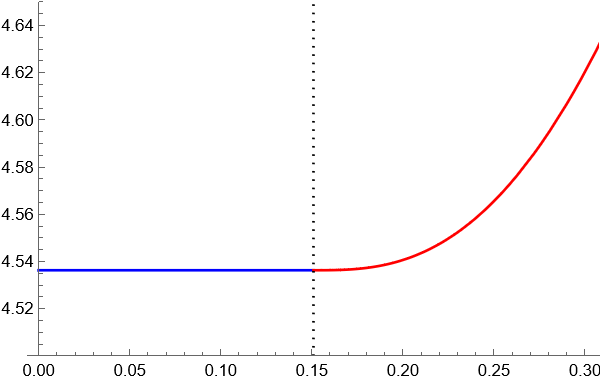}
    \subcaption{$Q_0=4, Q_1=2$}
  \end{minipage}   
    \caption{The graph shows the energy $w \mapsto  \frac{1}{N^2} \log K_N^{\rm post}$ when the condition \eqref{wdc} is violated, and the energy $w \mapsto  \frac{1}{N^2} \log K_N^{\rm pre}$ when  \eqref{wdc} holds. Here, plot (a) illustrates the case with symmetric charges, where $Q_0=Q_1=4$, while plot (b) illustrates the case with asymmetric charges, where $Q_0=4$ and $Q_1=2$. The vertical lines indicate the critical point \eqref{def of w critical point}. The two graphs (a) and (b) have a similar shape, but a difference in their values can be observed; specifically, the values decrease as \( Q_1 \) decreases. 
    }
    \label{Fig_energy Kpostpre}
\end{figure}

\begin{remark} ${}$ \\
1.~For application to (\ref{2.8c+})
the relevant choices of $Z$ in Proposition \ref{P8y} are $Z=\zeta(1)$ and $Z=w$, with the corresponding values
of $u_Z$ being $u_Z = 1$ and $u_Z=w$ respectively. \\
2.~Analogous to (\ref{2.97}) we have
\begin{equation}\label{2.97c}
\lim_{w \to \infty} W(z) = {(1 + Q_0 + Q_1) \over \pi}
 \int_{|z|<1/\sqrt{1/(Q_0+Q_1)}} 
 {\log |Z - z|^2 \over (1 + |z|^2)^2} 
 \, d^2 z  
 =  \log |Z|^2,
\end{equation}
where the second equality follows from the fact that $\log |Z - z|^2 $ can, for $|Z|$ outside of the droplet,
be replaced by $\log |Z|^2$ as follows from the multipole expansion and  rotational symmetry. With regards to (\ref{Xx91}), we see that the first three terms vanish in the limit $w \to \infty$. For the final term, in this limit the requirement that $Z= \zeta(u_Z)$ simplifies to $Z= R/ u_Z$, giving agreement with (\ref{2.97c}). \\
3.~Analogous to  the result of Remark \ref{R2.5}.2, we can simplify (\ref{Xx91}) in the case $Q_0 = Q_1$ with
$Q_0 = \mu /(w_s^2 - 1)$ for $\mu \to 1^+$ (i.e.~at the boundary with the post-critical phase). Thus, again with the aid of computer algebra, we find the simplified expressions as are relevant to (\ref{2.8c+})
\begin{align}\label{2.97d}
\lim_{\mu \to 1^+}{1 \over 2} W(\zeta(1)) & = - (Q_0 + 1) \log\Big ( {3 w_s^2 - 1 \over w_s^2 + 1} \Big ) -
Q_0 \log \bigg ( {1 \over 2} \Big (1 + {1 \over w_s^2} \Big ) \bigg ) + \log w_s, \nonumber \\
\lim_{\mu \to 1^+}{1 \over 2} W(\zeta(v_0)) & = - (Q_0 + 1) \log\Big (
{2 w_s^2 \over 1 + w_s^2} \big ) - Q_0 \log \bigg ( {1 \over 2} \Big (1 + {1 \over w_s^2} \Big ) \bigg ) + \log w_s,
\end{align}
where in this limit $Q_0 = 1/(w_s^2 - 1)$. \\
4.~Substituting (\ref{2.97d}) and the result of Remark \ref{R2.5}.2 in (\ref{2.8c+}), supplemented by the facts for in the setting of point 3.
\begin{equation}\label{2.97e}
\lim_{\mu \to 1^+}\log (1 + |\zeta(1)|^2) = \log \Big ( {(1 + w_s^2)^3 \over ( 3 w_s^2 - 1)^2}  \Big ), \quad 
\lim_{\mu \to 1^+}\log (w - \zeta(1))  = \log \Big ( {(1 + w_s^2)^2 \over 2 w_s  ( 3 w_s^2 - 1)}  \Big ),
\end{equation}
we are able to compute a simplified formula for (\ref{2.8c+}) in this circumstance. Thus we find
\begin{equation}\label{2.97f}
\lim_{\mu \to 1^+} \Big (\frac{4}{\beta N^2} \log \mathcal K_N^{\rm pre} \Big ) =
1+2Q_0 +2(1+2Q_0) \log(1+2Q_0) - 2 (1 + Q_0)^2 \log(1+Q_0) - 2Q_0^2 \log Q_0.
\end{equation}
This is identical to ${4 \over \beta N^2} \log \mathcal K_N^{\rm post}$ with $Q_0 = Q_1$ as given by
(\ref{C2}), which is necessary since the limit $\mu \to 1^+$ corresponds to the phase boundary. cf. Figure~\ref{Fig_energy Kpostpre}.
\\
5.~A companion to the simplified form (\ref{2.97f}), this time applying for general $Q_0, Q_1 \ge 0$, is the large $w$ expansion
\begin{multline}\label{2.97g}~
\frac{4}{\beta N^2} \log \mathcal K_N^{\rm pre}  \sim 2 Q_0 Q_1 \log |w|^2 + (Q_0+Q_1+1) +
 (Q_0 + Q_1) \log \Big ( {1 \over Q_0 + Q_1} \Big ) \\
 - (Q_0+Q_1+1)  (Q_0 + Q_1) \log\Big (1 +  {1 \over Q_0 + Q_1} \Big ) + \cdots,
 \end{multline}
 where higher order terms go to zero with $w$. This is deduced using (\ref{2.97}) and (\ref{2.97c}) in (\ref{2.8c+}). 
cf. \cite[Eq.(1.13)]{BKS23} for a comparison. See also the left graph in Figure~\ref{Fig_energy2D1D}.
\\
6.~A question of long standing interest in the study of the Ginibre ensembles and one-components plasmas more generally is the large $N$ asymptotic form of the probability that a region is free of eigenvalues (particles); see \cite[\S 3.1]{BF24}. The leading term in this asymptotic expansion is known to involve a weighted integration over the hole region, as well as over of a measure supported on its boundary (the so-called balayage measure) \cite{Ad18,AR17,Ch23}. The case of the spherical ensemble is considered explicitly in \cite[Th.~2.5(ii)]{Ch23} (see also \cite{BP24}), where the weighted integration over the hole region is shown to be precisely the LHS of (\ref{2.9s}). Another point of interest is that the balayage measure can be specified in terms of the conformal map from the unit disk to the hole region \cite[\S 4.2]{Ch23}.
\end{remark}

\section{Duality identity from random matrix theory and first consequences} 
 
 \subsection{The complex (generalised) spherical ensemble}
 
In random matrix theory, the complex (generalised) spherical ensemble of non-Hermitian random matrices is specified by matrices
of the form $G_1^{-1} G_2$, where $G_1, G_2$ are independent
 GinUE matrices (GinUE matrices are $N \times N$ matrices where each entry in an independent standard complex Gaussian);
 see \cite[\S 2.5]{BF24}.
 We denote by SrUE${}_{(N,K)}$ 
  the generalisation of the complex spherical ensemble defined by the construction $(G^\dagger G)^{-1/2} X$,
 where $G$ is an $(N+K) \times N$  complex standard Gaussian matrix, and $X$ is a GinUE matrix (the spherical ensemble
 corresponds to the case $K=0$). 
 
 From \cite[Exercises 3.6 q.3 with $\beta =2$ and $n_1 = N + K$]{Fo10}, a matrix $Y \in {\rm SrUE}_{(N,K)}$ has distribution
\begin{equation}
\det ( \mathbb I_N + Y^\dagger Y)^{-K - 2N}.
\end{equation}
It can be deduced from this that the eigenvalue probability density function of matrices from $ {\rm SrUE}_{(N,K)}$ is
proportional to (see e.g.~\cite[Eq.~(2.54) with $M=N$ and $n=K+N$]{BF24})
\begin{equation}\label{3}
\prod_{l=1}^N {1 \over (1 + | z_l|^2)^{K+N+1} }\prod_{1 \le j < k \le N} |z_k - z_j|^2.
\end{equation}
We recognise this as the case $r=0$ of the RHS of (\ref{13}).

\subsection{Duality identity}\label{S3.2}
Let ${\rm JUE}_{n,(a_1,a_2)}$ (the Jacobi unitary ensemble) refer to the random matrix ensemble with 
 eigenvalue PDF proportional to
\begin{equation}\label{1.1}
\prod_{l=1}^n  \lambda_l^{a_1} (1 - \lambda_l)^{a_2} \prod_{1 \le j < k \le n} | \lambda_k - \lambda_j|^2,
\end{equation}
supported on $0 < \lambda_l < 1$.
In the recent work \cite[Final equation of Prop.~5.4]{Fo24+}, the duality identity
\begin{equation}\label{1}
\langle | \det ( x \mathbb I_N - A X) |^{2r} \rangle_{X \in {\rm SrUE}_{N,K+r}} =
\Big \langle \prod_{l=1}^r \det  (  |x|^2 \mathbb I_N + t_l A A^\dagger  )    \Big \rangle_{\mathbf t \in {\rm JUE}_{r,(0,K-r)}}
\end{equation}
was obtained; cf.~(\ref{1Dz}). (The terminology duality relation comes about due to the 
interchange of the role of the key parameters
on each side.)

Our interest is in the case that $A =  \mathbb I_N$ and $x = w$ (the latter for notational convenience) when the RHS reduces to 
   \begin{equation}\label{1a}
\Big \langle \prod_{l=1}^r (  |w|^2 \mathbb I_N + t_l   )^N   \Big \rangle_{\mathbf t \in  {\rm JUE}_{r,(0,K-r)}}.
 \end{equation}
In terms of the eigenvalues of matrices from ${\rm SrUE}_{N,K}$,
the factor $ | \det ( x \mathbb I_N - A X) |^{2r}$ on the LHS of (\ref{1}) with $A =  \mathbb I_N$, $x=w$ is equal to
\begin{equation}\label{1aY}
\prod_{l=1}^N|w -   z_l |^{2r}.
\end{equation}
Hence the duality (\ref{1}) in this special case reads
\begin{equation}\label{1aX}
\Big \langle \prod_{l=1}^N|w -   z_l |^{2r}
\Big \rangle_{\mathbf z \in {\rm SrUE}_{N,K+r}} = \Big \langle \prod_{l=1}^r (  |w|^2 \mathbb I_N + t_l   )^N   \Big \rangle_{\mathbf t \in  {\rm JUE}_{r,(0,K-r)}};
\end{equation}
cf.~(\ref{1Dz}).
In relation to the LHS, significant from the viewpoint of the Coulomb gas system introduced in \S \ref{S1.1}, or more explicitly its stereographic projection to the plane as discussed in \S \ref{S2.1},
is that multiplying (\ref{1aY}) with (\ref{3}) gives the functional form
\begin{equation}\label{3a}
\prod_{l=1}^N { | w -  z_l |^{2r}  \over (1 + | z_l|^2)^{K+r+N+1} }  \prod_{1 \le j < k \le N} |z_k - z_j|^2,
\end{equation}
as appears on the RHS of (\ref{13}),  up to the factor of $(1 + |w|^2)^{-r(K+N)} $.

According to (\ref{1.1}), up to normalisation, the RHS of (\ref{1aX}) is equal to
  \begin{equation}\label{X}
  \int_0^1 dt_1 \cdots \int_0^1 dt_r \, \prod_{l=1}^r (|w|^2 + t_l)^N (1 - t_l)^{K-r} \prod_{1 \le j < k \le r} | t_k - t_j|^2,
 \end{equation}
 where it is required that $K - r > -1$. Changing variables $1 - t_l = s_l$, then $s_l \mapsto (1 + |w|^2) s_l$ gives that  (\ref{X})
 is equal to
 \begin{equation}\label{X1}    
 (1 + |w|^2)^{r(N+K)}
 \int_0^{1/(1 + |w|^2)} ds_1 \cdots  \int_0^{1/(1 + |w|^2)} ds_r  \, \prod_{l=1}^r (1 - s_l)^N s_l^{K-r} \prod_{1 \le j < k \le r} | s_k - s_j|^2.
 \end{equation}
 Removing the factor of $ (1 + |w|^2)^{r(N+K)}$ by associating its reciprocal with the functional form on the LHS of the duality (\ref{X}) --- this is precisely the factor that is otherwise missing in identifying
 the LHS of the duality with the RHS of (\ref{13}).
 The $r$-dimensional integral which remains in
 (\ref{X1}) is, after normalisation, equal to the gap probability $E(0,(1/(1 + |w|^2),1);  {\rm JUE}_{r,(K-r,N)})$.
 Here we have made use of notation analogous to that used below (\ref{1Ex}). Specifically, this quantity specifies the probability that, for the matrix
 ensemble ${\rm JUE}_{r,(K-r,N)}$ there are no eigenvalues in the interval $(1/(1 + |w|^2,1)$. Denote by $\mathcal Q_N(Q_0,Q_1,w)$ the RHS of (\ref{13}) with $r$ and $K$ specialised as in (\ref{sf}) and integrated over $\mathbb C^N$ (this
 then is the configuration integral associated with the partition function for the Coulomb gas).
 In terms of the gap probability we then have the rewrite of the duality identity
 \begin{equation}\label{X9} 
 {\mathcal Q_N(Q_0,Q_1,w) \over \mathcal Q_N(Q_0,0,\cdot)} = C_{N,Q_0N,Q_1N} \,
E(0,(1/(1 + |w|^2),1);  {\rm JUE}_{Q_1 N ,((Q_0 - Q_1)N,N)}),
\end{equation}
where
 \begin{equation}\label{X10} 
C_{N,Q_0N,Q_1N} = {J_{Q_1 N ,(0,(Q_0 - Q_1)N)} \over J_{Q_1 N ,((Q_0 - Q_1)N, N)}}, \quad
J_{n,(a,b)} := \int_0^1 dt_1 \cdots \int_0^1 dt_n \, \prod_{l=1}^n t_l^a (1 - t_l)^b
\prod_{1 \le j < k \le n} (t_k - t_j)^2.
\end{equation}

We remark that an analogous rewrite of the duality (\ref{1Ey}) was given in \cite{BSY24}, and it was from this rewrite that the asymptotic expansions of the partition function for the one component plasma with a macroscopic point charge insertion (\ref{2Ex}), (\ref{2E}), (\ref{2EA}) were deduced. An opportunity for similar analysis is provided by a duality identity
relating the averaged characteristic polynomial for the eigenvalues of the random matrix ensemble formed by a sub-block of a random Haar unitary matrix  to a gap probability in a particular Jacobi ensemble \cite{DS22,SSD23,SS24,Se23}; see Appendix C below. The point process corresponding to the eigenvalues of a sub-block of a random Haar unitary matrix can be identified with a one-component Coulomb gas model in the Poincar\'e disk \cite{FK09}, or equivalently on the surface of the so-called pseudo-sphere (defined in hyperbolic geometry and which has a constant negative curvature) \cite{JT98}. A study of the associated potential problem for a macroscopic point insertion in this model has been undertaken in the recent work \cite{BCMS25}.

 \subsection{Proof of Proposition \ref{P1.1}}\label{S3.3}
 In accordance with the definition (\ref{1E}), our interest is in the integral over the spherical domains
 on the LHS of (\ref{13}) in the large $N$ limit.  This multidimensional
 integral can be studied using the duality identity (\ref{1aX}).
 Furthermore, we want to study this limit with $K, r$ as in (\ref{sf}).
  
 To study the large $N$ limit we can consider the RHS of  (\ref{X9}). 
 Relevant is the density of Jacobi ensemble ${\rm JUE}_{n,(n \gamma_1,n \gamma_2)}$
 in the large $n$ limit.   We note that the requirement in (\ref{X9}) that $Q_0 \ge Q_1$ does not restrict
 our consideration of Proposition \ref{P1.1} to this case, as by rotational invariance of the sphere, we
 can always choose to have the biggest charge at the south pole.
The Jacobi density is given by a functional form first identified by
 Wachter \cite{Wa78}
   \begin{equation}\label{Wa1}
 \rho^{\rm J}(x)  =  (\gamma_1 + \gamma_2+2) {\sqrt{(x - c^{\rm J}) ( d^{\rm J} - x)} \over 2 \pi x (1 - x)}
  \end{equation}   
  supported on $(c^{\rm J}, d^{\rm J})$ with these endpoints specified by
   \begin{equation}\label{Wa2}   
 {1 \over (\gamma_1 + \gamma_2 + 2)^2} \Big ( \sqrt{(\gamma_1 + 1) (\gamma_1 + \gamma_2 + 1)} \pm \sqrt{\gamma_2 + 1} \Big )^2.
  \end{equation} 
(See also \cite{MSV79} and \cite[Corollary 4.1]{Co05} for the appearance of (\ref{Wa1}) in  related  contexts.)
  Comparing with the RHS of (\ref{X9}) we have $n = r = Q_1 N$, $n \gamma_1  = K - r = (Q_0 - Q_1) N$,  $ n \gamma_2 = N $,
   and thus
    \begin{equation}\label{gg} 
\gamma_1 = (Q_0 - Q_1)/Q_1,   \quad \gamma_2 = 1/Q_1.
   \end{equation} 
  Outside of this interval, the density is exponentially small with respect to $N$ \cite{Fo12}.
  
  Significant is the situation that
   \begin{equation}\label{Y}
   1/(1 + |w|^2) > d^{\rm J}
   \end{equation}  
   i.e.~$ 1/(1 + |w|^2) $ is greater than the right hand endpoint of the support. 
   Then, due to the functional dependence on $N$ for $x > d^{\rm J}$,
    up to terms exponentially small in $N$ we have that for large $N$
    \begin{equation}\label{Y1}   
  E(0,(1/(1 + |w|^2),1); {\rm JUE}_{_{Q_1 N ,((Q_0 - Q_1)N,N)}}) = 1 +{\rm{O}}(e^{ -\epsilon }),
  \end{equation}  
  for some $\epsilon>0.$ In this case we do not need to do any more calculation relating to (\ref{X1}), and moreover our task of computing the normalisation of the Boltzmann
  factor is seen to be independent of $w$. This allows us  to take $w=0$, which in the Coulomb gas on a sphere picture corresponds to placing
  the charge $Q_1N$ at the north pole. This is the circumstance considered in \cite{FF11}. Thus we can read off the result (\ref{1G})
  (which follows by inserting the
  various normalisations) from \cite[Prop.~4.1]{FF11}; the statement in relation to the structure of the $1/N$ expansion follows from the asymptotics of the Barnes $G$-function as carried out explicitly in \cite[derivation of Eqns.~(3.6) and (3.7)]{BSY24}.   
  
  \subsection{Critical regime}
  We consider next the situation that to leading order for large $r$ the inequality (\ref{Y}) becomes an equality.
  In analogy with (\ref{2F}) introduce a sub-leading correction proportional to $N^{-2/3}$ by requiring that
   \begin{equation}\label{2Fa}
    1/(1 + |w|^2) =  d^{\rm J}    + {s \over \alpha(Q_1,Q_2) N^{2/3}}.
  \end{equation}      
 To specify $\alpha(Q_1,Q_2)$ herein, with $\gamma_1, \gamma_2$ as in (\ref{gg}), introduce
    \begin{equation}\label{2Fb}
  \mathfrak c^2 = {\gamma_1 + 1 \over \gamma_1 + \gamma_2 + 2}, \quad
  \tilde {\mathfrak c}^2 = {1 \over \gamma_1 + \gamma_2 + 2}, \quad
   \mathfrak s^2 = {\gamma_2 + 1 \over \gamma_1 + \gamma_2 + 2}, \quad
  \tilde {\mathfrak s}^2 = {\gamma_1 + \gamma_2 + 1 \over \gamma_1 + \gamma_2 + 2}.
  \end{equation}
  One then has \cite[Eq.~(7)]{HF12}
      \begin{equation}\label{2Fc}
\alpha(Q_0,Q_1) = {1 \over       \mathfrak c  \mathfrak s   \tilde {\mathfrak c}   \tilde {\mathfrak s}}
\Big ( {  \mathfrak c  \mathfrak s   \tilde {\mathfrak c}   \tilde {\mathfrak s} (\gamma_1 + \gamma_2 + 2)^{1/2} \over 
 \tilde {\mathfrak c}   \tilde {\mathfrak s}  (  \mathfrak c^2  -   \mathfrak s^2 ) +   \mathfrak c  \mathfrak s 
( \tilde {\mathfrak c}^2 -     \tilde {\mathfrak s}^2)} \Big )^{4/3}.
  \end{equation}
  The significance of (\ref{2Fa}) is that it replacing $1/(1+|w|^2)$ in the LHS of (\ref{Y1}) allows for the $\beta = 2$ soft edge limit
law  formula \cite{HF12}
   \begin{equation}\label{2Fd} 
   \lim_{N \to \infty} E\Big (0,(d^{\rm J}    + s / ( \alpha(Q_0,Q_1) N^{2/3}),1  ); {\rm JUE}_{r,(K-r,N)} \Big ) \Big |_{K=Q_0 N, r = Q_1 N} =
  E_2^{\rm soft}(0;(Q_0^{-2/3}s,\infty)); 
   \end{equation}
in relation to the RHS recall (\ref{2F1}). 

\begin{cor}
With $w$ dependent on $N$ as specified by (\ref{2Fa}), the partition function (\ref{1E}) has for its large $N$ expansion
the RHS of (\ref{1F}), with the addition of $\log E_2^{\rm soft}(0;(Q_0^{-2/3}s,\infty))$, corresponding to a further term of
${\rm O}(1)$.
\end{cor}

\section{Electrostatic consequences of the duality in the pre-critical regime}\label{S4t}

\subsection{Preliminaries}
From the discussion of \S \ref{S3.3} we have that in the so called post-critical regime, when
the inequality (\ref{1F}) holds, the duality identity (\ref{X9}) gives that for large $N$
\begin{equation}\label{Y8}
\mathcal Q_N^{\rm post}(Q_0,Q_1,w) = C_{N,Q_0N,Q_1N}  \mathcal Q_N(Q_0,0,\cdot))  \Big ( 1 + {\rm O}(e^{-\epsilon N})
\Big ),
\end{equation}
for some $\epsilon > 0$, i.e.~the gap probability factor can effectively we set equal to unity.
We take (\ref{Y8}) as the meaning of $\mathcal Q_N^{\rm post}(Q_0,Q_1,w)$ in the full parameter range
beyond (\ref{1F}).
Comparison with (\ref{X9}) without this simplifying the gap probability factor then shows that for
the parameter range of the pre-critical phase
\begin{equation}\label{Y9}
{\mathcal Q_N^{\rm post}(Q_0,Q_1,w) \over \mathcal Q_N^{\rm pre}(Q_0,Q_1,w) }
=  E(0,(1/(1 + |w|^2),1);  {\rm JUE}_{Q_1 N ,((Q_0 - Q_1)N,N)}) \Big ( 1 + {\rm O}(e^{-\epsilon N})
\Big ).
\end{equation}

In the present work we are restricting attention to the leading order large $N$ terms in the pre-critical regime. The Coulomb gas on a sphere viewpoint gives the leading large  $N$ forms
\begin{equation}\label{Y9a}
\log \mathcal Q_N^{\rm post}(Q_0,Q_1,w) = - 2  N^2  K_N^{\rm post} + \cdots , \quad
\log \mathcal Q_N^{\rm pre}(Q_0,Q_1,w) = - 2  N^2  K_N^{\rm pre} + \cdots,
\end{equation}
where $K_N^{\rm post} = - {1\over  \beta  N^2} \log  \mathcal K_N^{\rm post}$ as specified by (\ref{C2}), while $K_N^{\rm pre} = - {1  \over \beta N^2} \log  \mathcal K_N^{\rm pre}$ as specified by
(\ref{2.8c+}); in particular both $K_N^{\rm post}, K_N^{\rm pre}$ are independent of $N$. For consistency with (\ref{Y9}), we must then have the leading large $N$ form
of the gap probability
\begin{equation}\label{Y9b}
\log E(0,(1/(1 + |w|^2),1);  {\rm JUE}_{Q_1 N ,((Q_0 - Q_1)N,N)}) = 2N^2 ( K_N^{\rm pre} - K_N^{\rm post}  
 ) +  \cdots;
\end{equation}
note that the RHS must be negative for parameter values in the pre-critical regime since the LHS is the logarithm of a probability. This is consistent with the equilibrium measure having the property that it minimises the energy functional corresponding to $K_N^{(\cdot)}$ in the particular parameter regime. In the context
of probability theory, such a functional form is characteristic of the large deviation regime;
see \cite{RKC12} for the present context.
In fact the leading large $N$ form of LHS of (\ref{Y9b}) in the pre-critical regime can be
computed by electrostatic arguments particular to the log-gas interpretation of
the Jacobi unitary ensemble. This is to be discussed in the next subsection. It gives rise to an identity between
electrostatic energies of distinct Coulomb gas systems, which moreover involves distinct functional forms, see \eqref{relation 2D 1D energies}.

\subsection{Comparative electrostatic energies}

Let us first introduce the constrained spectral density of the JUE. This is the limiting spectral density of ${\rm JUE}_{n,(n \gamma_1, n \gamma_2)}$, with the requirement that its maximal eigenvalue is smaller than a given $\zeta \in (0, d^{\rm J})$, where $d^{\rm J}$ is specified in \eqref{Wa2}. 
In general, for one-dimensional log gases, if one imposes such a hard wall constraint that overlaps with the associated equilibrium measure, it causes macroscopic changes, resulting in a new equilibrium measure. Nonetheless, this new measure remains absolutely continuous with respect to the one-dimensional Lebesgue measure. This phenomenon contrasts with two-dimensional Coulomb gases, where the new measure (given in terms of the balayage measure \cite{Ch23}) under the hard wall constraint is singular with respect to the two-dimensional Lebesgue measure.

\begin{figure} 
 \begin{minipage}[b]{0.48\linewidth}
    \centering
    \includegraphics[width=\textwidth]{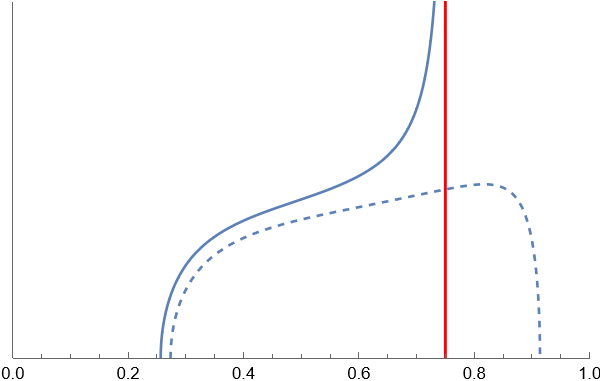}
    \subcaption{$\rho^{\rm J}(x)$ and  $\rho^{\rm J}(x; (0, \zeta))$}
  \end{minipage} \quad 
   \begin{minipage}[b]{0.48\linewidth}
    \centering
    \includegraphics[width=\textwidth]{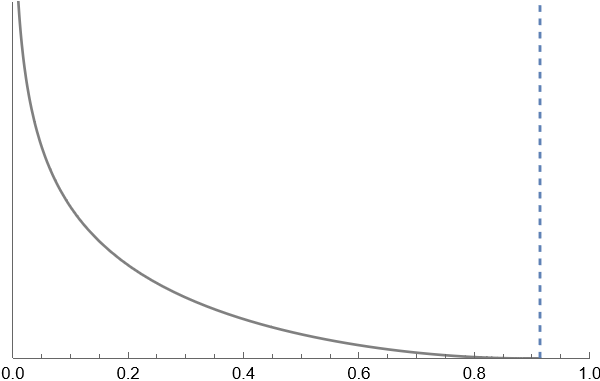}
    \subcaption{$S(L(\zeta),\zeta)-S(c^{ \rm J },  d^{ \rm J } )$}
  \end{minipage}  
    \caption{The plot (a) illustrates the Wachter distribution \( \rho^{\rm J}(x) \) (dashed line) alongside the constrained density \( \rho^{\rm J}(x; (0, \zeta)) \) (solid line), subject to the hard wall constraint \( \zeta = 0.75 \) (vertical line). Here, \( \gamma_1 = 4 \) and \( \gamma_2 = 2 \), resulting in \( c^{\rm J} \approx 0.274 \), \( d^{\rm J} \approx 0.914 \), and \( L(\zeta) \approx 0.247 \).  
The plot (b) shows the graph \( \zeta \mapsto \big( S(L(\zeta), \zeta) - S(c^{\rm J}, d^{\rm J}) \big) \) with the same parameter choices. The vertical line indicates the right edge \( d^{\rm J} \). } 
    \label{Fig_JUE constrained}
\end{figure}

For purposes of describing the constrained spectral density, it is convenient to write
\begin{align*}
\mathsf{B}= \frac{ \gamma_1 \sqrt{\zeta} }{ 2+\gamma_1+\gamma_2 }, \quad \mathsf{C}= \frac{\gamma_2 \sqrt{1-\zeta} }{ 2+\gamma_1+\gamma_2 }.  
\end{align*}
Set
\begin{align*}
\mathsf{Q}= -\frac19 (1-\mathsf{B}^2-\mathsf{C}^2)^2,
\quad 
\mathsf{R} = \frac{1}{27} \Big( \mathsf{B}^6-3\mathsf{B}^4(1-\mathsf{C}^2)+3\mathsf{B}^2(1+16\mathsf{C}^2+\mathsf{C}^4)-(1-\mathsf{C}^2)^3\Big) .
\end{align*}
These are building blocks to define 
\begin{equation}
L(\zeta):= \frac14 \bigg( \mathsf{B}+\sqrt{z_0} -\sqrt{ \mathsf{B}^2-2\mathsf{C}^2+2 -z_0-\frac{2}{ \sqrt{z_0}  }  \mathsf{B}(\mathsf{C}^2+1) } \bigg)^2, 
\end{equation}
where 
\begin{align*}
z_0= -\frac{2\mathsf{C}^2-\mathsf{B}^2-2 }{3} +( \mathsf{R}+\sqrt{\mathsf{Q}^3+\mathsf{R}^2} )^{ 1/3 } +( \mathsf{R}-\sqrt{\mathsf{Q}^3+\mathsf{R}^2} )^{ 1/3 }. 
\end{align*}
This value \( L(\zeta) \) corresponds to the left edge of the new equilibrium measure. In particular, \( L(\zeta) \leq c^{\rm J} \), indicating that the JUE support under the constraint is shifted towards the left.  
More precisely, it was shown in \cite{RKC12} that the limiting density of \( {\rm JUE}_{n, (n \gamma_1, n \gamma_2)} \), conditioned on having its maximal eigenvalue smaller than $\zeta$, is given by
\begin{equation} \label{Wa constrained}
\rho^{\rm J}(x;(0,\zeta)) =  (\gamma_1+\gamma_2+2) \sqrt{ \frac{x-L(\zeta)}{\zeta-x} } \frac{1}{2\pi x(1-x)} \Big( \frac{\gamma_1\, \sqrt{  \zeta/ L(\zeta)   }}{2+\gamma_1+\gamma_2} -x \Big) 
\end{equation}
supported on $(L(\zeta),\zeta)$.  
In particular, if $\zeta=d^{ \rm J }$, it follows from \eqref{Wa2} that 
\begin{align*}
\rho^{\rm J}(x;(0,\zeta)) \Big|_{ \zeta=d^{ \rm J } } = \rho^{\rm J}(x). 
\end{align*}
This agrees with the intuition that if the hard wall constraint does not overlap with the equilibrium measure, then the constraint does not cause a macroscopic change to the ensemble.

The resulting density \eqref{Wa constrained} can be used to derive the leading order asymptotic behaviour of the probability that the maximal eigenvalue of \( {\rm JUE}_{n, (n \gamma_1, n \gamma_2)} \) is smaller than $\zeta$. By computing the logarithmic energy associated with \eqref{Wa constrained}, it was derived in \cite{RKC12} that this probability is asymptotically given by 
\begin{equation}
\exp \Big( -n^2 ( S(L(\zeta),\zeta)-S(c^{ \rm J },  d^{ \rm J } ) ) + {\rm o}(n^2)\Big),
\end{equation}
where 
\begin{align}
\begin{split}
S(x,y) &= -(\gamma_1+\gamma_2+2) \Big( \gamma_1 \log \frac{\sqrt{x}+\sqrt{y}}{ 2 } +  \gamma_2 \log \frac{\sqrt{1-x}+\sqrt{1-y}}{ 2 }  \Big)
\\
&\quad + \frac{\gamma_1^2}{4}\log(xy) + \frac{\gamma_2^2}{4}\log\Big((1-x)(1-y) \Big) 
\\
&\quad + \gamma_1 \gamma_2 \log \frac{ \sqrt{x(1-y)}+\sqrt{y(1-x)} }{ 2} - \log\frac{y-x}{4}.
\end{split}
\end{align}
We mention that for more general one-dimensional log gases, the precise asymptotic behaviour can also be derived using the Riemann-Hilbert analysis \cite{CG21}, cf. \cite{BCMS25a}. In addition, when $\zeta \to d^{ \rm J }$, the large deviation rate function has a cubic-decay, which is often referred to as a third order phase transition, see e.g. \cite{MS14} and references therein.  
Notice from \eqref{Y9b} that
\begin{equation} 
\log E(0,(1/(1 + |w|^2),1);  {\rm JUE}_{ N ,( \frac{Q_0}{Q_1} - 1)N,\frac{1}{Q_1}N)}) = \frac{2}{Q_1^2}N^2 ( K_N^{\rm pre} - K_N^{\rm post}  ) +  \cdots;
\end{equation}
This leads to the relation 
\begin{equation} \label{relation 2D 1D energies}
\frac{2}{Q_1^2}  ( K_N^{\rm post}  - K_N^{\rm pre} ) =  \Big( S(L(\zeta),\zeta)-S(c^{ \rm J },  d^{ \rm J } ) \Big)\bigg|_{ \gamma_1=\frac{Q_0}{Q_1}-1, \gamma_2 =\frac{1}{Q_1},\zeta=\frac{1}{1+|w|^2} } . 
\end{equation}
See Figure~\ref{Fig_energy2D1D} for the numerics.

\begin{figure}  
    \includegraphics[width=0.48\textwidth]{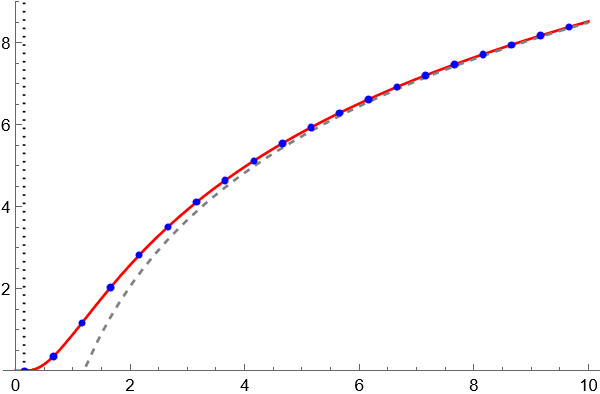}
   \quad 
    \includegraphics[width=0.48\textwidth]{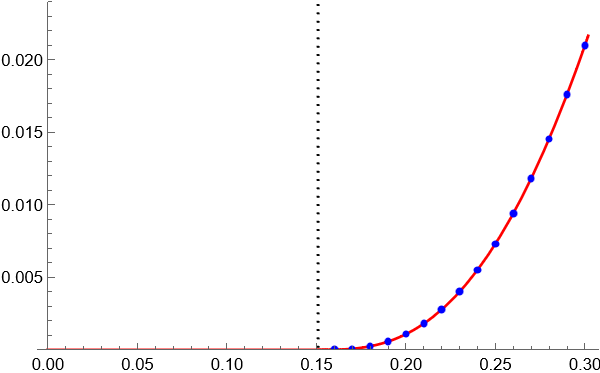} 
    \caption{ The blue dots represent $w \mapsto \frac{2}{Q_1^2}  ( K_N^{\rm post}  - K_N^{\rm pre} )$, where $Q_0=4, Q_1=2$. The solid red curve shows the RHS of \eqref{relation 2D 1D energies}, with $\gamma_1=1$, $\gamma_2=1/2$ as specified. The vertical line marks the critical transition point defined in \eqref{def of w critical point}. In the left figure, where \( w \) ranges from small to large values, the dashed curve represents the RHS of \eqref{2.97g}, scaled by an appropriate multiplicative constant, demonstrating that the asymptotic behaviour is in good agreement for large $w$. 
    } 
    \label{Fig_energy2D1D}
\end{figure}

\subsection*{Acknowledgements}
 Sung-Soo Byun was supported by the POSCO TJ Park Foundation (POSCO Science Fellowship), by the New Faculty Startup Fund at Seoul National University and by the LAMP Program of the National Research Foundation of Korea (NRF) grant funded by the Ministry of Education (No. RS-2023-00301976).
Funding support to Peter Forrester for this research was through the Australian Research Council Discovery Project grant DP250102552. 
Sampad Lahiry acknowledges financial support from the International Research Training Group (IRTG) between KU Leuven and University of Melbourne and Melbourne Research Scholarship of University of Melbourne.
 
\appendix
\section*{Appendix A}\label{A1}
\renewcommand{\thesection}{A} 
\setcounter{equation}{0}
Associate with the south pole a spherical cap of area $\pi Q_0/(Q_0 + Q_1 + 1)$. 
  Denote the angle determining the boundary of this spherical cap by $\theta_0$. In terms of this angle,
  the corresponding Cayley-Klein parameters on the boundary at $\phi=0$, $(u_0, v_0)$ say,
  are such that $ \sin^2 {\theta_0 \over 2} =
   | v_0 |^2$ and so
    \begin{equation}\label{Y1x}  
  | v_0 |^2 = \sin^2 {\theta_0 \over 2} =
  {Q_0     \over Q_0 + Q_1 + 1},
  \end{equation}  
  where we have used the fact that the area of the spherical cap is equal to $\pi \sin^2 {\theta_0 \over 2}$.
  Denoting by $r_{Q_0}$ the point on the positive real axis that $(u_0, v_0)$ is mapped to under the
  stereographic projection, it follows that 
    \begin{equation}\label{Y2}    
  r_{Q_0}^2 =    {Q_1 +1 \over Q_0},
   \end{equation}  
as previously noted in \cite[second equation in (2.24)]{FF11}.

Let the point $\mathbf a_w$ on the sphere with Cayley-Klein parameters $(u_w,v_w)$ have spherical coordinates
$(\theta,\phi) = (\theta_w,0)$.  Associate with this point a spherical cap of area
$\pi Q_1/(Q_0 + Q_1 + 1)$. Let the polar coordinates of the extremities of the spherical cap about $(\theta_w,0)$ be given by
  $\theta_{w} - \hat{\theta} < \theta < \theta_{w} + \hat{\theta}$. By rotation invariance, we can take $\theta_{w} = 0$ allowing
  us to conclude
    \begin{equation}\label{Y1a}  
   \sin^2 {\hat{\theta} \over 2} =
  {Q_1     \over Q_0 + Q_1 + 1}.
  \end{equation}  
  Let the extremity of this cap in the direction of the south pole, and restricted to $\phi=0$,  be denoted $\mathbf a_+$,
  and its stereographic projection to the positive real axis be denoted $|w_+|$.
  It follows from (\ref{Y1a}) and  (\ref{SP1}) that
   \begin{equation}\label{Y1c}  
   |\mathbf a_+ - \mathbf a_w|^2 =  {Q_1     \over Q_0 + Q_1 + 1} = {1 \over (1 + |w|^2)} {1 \over (1 + |w_+|^2)} (|w_+| - |w|)^2.
  \end{equation}  
  
  For the spherical cap associated with the south pole not to overlap the spherical cap associated
  with $\mathbf a_w$, we require
   \begin{equation}\label{Y1d}    
   |w_+|^2 <  r_{Q_0}^2 =    {Q_1 +1 \over Q_0}, 
    \end{equation} 
  where the equality is
  (\ref{Y2}). We use this inequality as an equality in (\ref{Y1c}), leaving a quadratic equation for $|w|$.
 After some manipulation we can verify that the
 critical case of  (\ref{1F}), when the inequality is an equality, is a solution for $1/(1+|w|^2)$.

 \appendix
\section*{Appendix B} \label{B1}
\renewcommand{\thesection}{B} 
\setcounter{equation}{0}
Denote the fourth order polynomial in $\alpha$ on the LHS of (\ref{abd3}) by $p(\alpha)$. Here we will outline calculations, based on the discriminant Disc${}_\alpha(p)$, which show that for $w > 0$ this polynomial has all roots real and distinct, two of which are positive and two negative.

First, use of computer algebra shows
 \begin{equation}
 {\rm Disc}_{\alpha}(p) = C (1 + w^2) h(w^2),
  \end{equation} 
  where $h(x)$ is a degree three polynomial in $x$, and $C > 0$. Thus all
  roots are real. Next we will show that all the roots are distinct, which requires showing that $h(x)$ has
  no positive roots. The polynomial $h(x)$ has a positive constant term and positive coefficient of $x^3$, implying that there is one negative root. A further computer algebra computation shows $ {\rm Disc}_{x}(h)<0$. Hence the remaining two roots of $h$ are complex conjugate pairs. Thus we have the required result that roots of $h$ are all positive, and so all roots of $p$ are distinct. On the other hand, in the special case $Q_0 = Q_1$ we know from the factorisation (\ref{abd3}) that there are exactly two positive and two negative roots, a situation which then must persist, considering too that for $w \ne 0$, $\alpha = 0$ is not a root.

 \appendix
\section*{Appendix C}  
\renewcommand{\thesection}{C} 
\setcounter{equation}{0}

We discuss a seemingly non-trivial connection with the truncated unitary ensemble.
We consider the top-left $N \times N$ submatrix ${\rm{T}}_{N,M}$ of a random unitary $M \times M$ matrix picked with respect to the normalised Haar measure. 
The PDF for the eigenvalues $\{ \lambda_j \}_{j=1}^N$ of ${\rm{T}}_{N,M}$ is given by \cite{ZS99}
\begin{equation} \label{def of tUE}
  \frac{1}{N!} \prod_{l=0}^{N-1} \frac{ \Gamma(M-N+l+1) } {l! \, \Gamma(M-N)}  \prod_{l=1}^N \frac{ 1  }{  (1 - | z_l|^2)^{N-M+1} }   \prod_{1 \le j < k \le N} |z_k - z_j|^2 \, d\mathbf r_1 \cdots d \mathbf r_N \qquad |\lambda_j| \le 1.  
\end{equation} 

It was shown in \cite[Theorem 1.10]{DS22} that for $r \in \mathbb Z^+$
\begin{align} 
 \label{TUE JUE relation}
\Big \langle  | \det ( {\rm{T}}_{N,M}-v )|^{2r } \Big \rangle  &=  \Big \langle | \det ( {\rm{T}}_{N,M} )|^{2r } \Big \rangle  (1-|v|^2)^{ -r(M-N+r)  }  \nonumber
\\
&\quad \times E(0,1-|v|^2,1);  {\rm JUE}_{r,(M-N,N)})   . 
\end{align} 
Further, by \cite[Eq.(1.24)]{DS22}, we have 
	\begin{align}  
		\begin{split}
			 \Big \langle | \det ( {\rm T}_{N,M} ) |^{2r} \Big \rangle  &= \prod_{l=0}^{N-1} \frac{ \Gamma(r+l+1) \Gamma(M-N+l+1) }{ l! \, \Gamma(M-N+r+l+1) },   
		\end{split}
	\end{align}
and so
\begin{multline}\label{C4}
\frac{1}{N!}  \int_{|\mathbf r_1| < 1} d\mathbf r_1 \cdots  \int_{|\mathbf r_N| < 1} d\mathbf r_N \,
\prod_{l=1}^N \frac{ | v -  z_l |^{2r}  }{  (1 - | z_l|^2)^{N-M+1} }  \prod_{1 \le j < k \le N} |z_k - z_j|^2
\\
=  (1-|v|^2)^{ -r(M-N+r)  } \, E(0,(1-|v|^2,1);  {\rm JUE}_{r,(M-N,N)})  \prod_{l=0}^{N-1}  
     \frac{ \Gamma(M-N) \Gamma(r+l+1)   }{   \Gamma(M-N+r+l+1) }   .   
\end{multline}

On the other hand
\begin{multline}
 {1 \over  (1 + |w|^2)^{r(K+N)} }  \frac{1}{N!}  \int_{|\mathbf r_1|<1} d\mathbf r_1 \cdots  \int_{|\mathbf r_N|<1} d\mathbf r_N  \,
\prod_{l=1}^N { | w -  z_l |^{2r}  \over   (1 + | z_l|^2)^{K+r+N+1} } \prod_{1 \le j < k \le N} |z_k - z_j|^2 
\\
 \propto E(0,(1/(1 + |w|^2),1);  {\rm JUE}_{r,(K-r,N)}), 
\end{multline}
as follows from the working in the paragraph below (\ref{1}).
Comparison with the RHS of (\ref{C4}) shows that under the setting of the parameters
\begin{equation}
|v|= \frac{|w|}{ \sqrt{1+|w|^2} }, \qquad K+N=M+r,
\end{equation}
the partition functions for the Coulomb gas systems corresponding to the  averaged power of the
characteristic polynomial in truncated unitary ensemble, and in the (generalised) spherical ensemble,
agree up to normalisations. 


\providecommand{\bysame}{\leavevmode\hbox to3em{\hrulefill}\thinspace}
\providecommand{\MR}{\relax\ifhmode\unskip\space\fi MR }
\providecommand{\MRhref}[2]{%
  \href{http://www.ams.org/mathscinet-getitem?mr=#1}{#2}
}
\providecommand{\href}[2]{#2}


\begin{thebibliography}{10}

\bibitem{Ad18}
K. Adhikari, \emph{Hole probabilities for $\beta$-ensembles and determinantal point processes in the complex plane}, Electron.
J. Probab. \textbf{23} (2018), Paper No. 48, 21 pp.

\bibitem{AR17}
K. Adhikari and N.K. Reddy, \emph{Hole probabilities for finite and infinite Ginibre ensembles}, Int. Math. Res. Not.
\textbf{2017} (2017), 6694--6730.

\bibitem{ABK21} G. Akemann, S.-S. Byun and N.-G. Kang, \emph{A non-Hermitian generalisation of the Marchenko–Pastur distribution: from the circular law to multi-criticality}, Ann. Henri Poincaré \textbf{22} (2021), 1035--1068.

\bibitem{ACC23}
Y. Ameur, C. Charlier and J. Cronvall, \emph{Free energy and fluctuations in the random normal matrix model with spectral gaps},
arXiv:2312.13904.


  \bibitem{AKS21}
Y. Ameur, N-G. Kang and S.-M. Seo,
\emph{The random normal matrix model: insertion of a point charge},
Potential Anal. \textbf{58}
(2023), 331--372.

\bibitem{BBLM15}
 F. Balogh, M. Bertola, S.-Y. Lee and K.D.T.-R. McLaughlin, \emph{Strong asymptotics of the orthogonal polynomials with respect to a measure supported on the plane}, Comm. Pure Appl. Math. \textbf{68} (2015), 112--172.

\bibitem{BM15}
F. Balogh and M. Merzi, \emph{Equilibrium measures for a class of potentials with discrete rotational symmetries}, Constr. Approx.
\textbf{42} (2015), 399--424.

\bibitem{BC12} F. Benaych-Georges and F. Chapon, \emph{Random right eigenvalues of Gaussian quaternionic matrices}, Random Matrices Theory Appl. \textbf{1} (2012), 1150009.

\bibitem{BDSW18}
J.S. Brauchart, P.D. Dragnev, E.B. Saff and R.S.~Womersley, 
\emph{Logarithmic and Riesz equilibrium for multiple sources on the sphere: the exceptional case}, Contemporary
Computational Mathematics - A Celebration of the 80th Birthday of Ian Sloan (J. Dick, F. Kuo, and H. Wo\'zniakowski, eds.), Springer, Cham, 2018, pp. 179--203.

\bibitem{By24} S.-S. Byun, \emph{Planar equilibrium measure problem in the quadratic fields with a point charge}, Comput. Methods Funct. Theory \textbf{24} (2024), 303--332.

\bibitem{BCMS25}
S.-S. Byun, C.~Charlier, P.~Moreillon and N.~Simm, \emph{Planar equilibrium measure for the truncated ensembles with a point charge}, manuscript in progress.

\bibitem{BCMS25a} S.-S. Byun, C.~Charlier, P.~Moreillon and N.~Simm, \emph{Duality in random matrix theory, Coulomb gases and last passage percolation}, manuscript in progress. 

\bibitem{BF23a} S.-S. Byun and P.J. Forrester, \emph{Spherical induced ensembles with symplectic symmetry}, SIGMA Symmetry Integrability Geom. Methods Appl. 19 (2023), 033, 28pp.

 \bibitem{BF24}
S.-S.~Byun and P.J.~Forrester,   \emph{Progress on the study of the Ginibre ensembles}, KIAS Springer Series in Mathematics \textbf{3}, Springer, 2024.

\bibitem{BKS23}
S.-S. Byun, N.-G. Kang and S.-M. Seo, \emph{Partition functions of determinantal and Pfaffian Coulomb gases with radially symmetric potentials}, Comm. Math. Phys. \textbf{401} (2023), 1627--1663.

\bibitem{BP24} S.-S. Byun and S. Park, \emph{Large gap probabilities of complex and symplectic spherical ensembles with point charges}, arXiv:2405.00386.  

 \bibitem{BSY24}
S.-S. Byun, S.-M. Seo, and M. Yang, \emph{Free energy expansions of a conditional GinUE and large deviations of the smallest eigenvalue of the LUE}, arXiv:2402.18983.


\bibitem{Ch23} C. Charlier, \emph{Hole probabilities and balayage of measures for planar Coulomb gases}, arXiv:2311.15285.  

\bibitem{CG21} C. Charlier and R. Gharakhloo, \emph{Asymptotics of Hankel determinants with a Laguerre-type or Jacobi-type potential and Fisher–Hartwig singularities} Adv. Math. \textbf{383} (2021), 107672.

\bibitem{Co05} B. Collins, \emph{Product of random projections, Jacobi ensembles and universality problems arising from free probability}, Probab. Theory Relat. Fields \textbf{133} (2005), 315--344.

 \bibitem{CK19}
J.G. Criado del Rey and A.B.J. Kuijlaars, \emph{An equilibrium problem on
the sphere with two equal charges}, preprint arXiv:1907.0480.

\bibitem{CK22} 
J.G. Criado del Rey and A.B.J. Kuijlaars, \emph{A vector equilibrium problem for symmetrically located point charges on a sphere}, Constr. Approx. \textbf{55} (2022), 775--827.

 \bibitem{CC03}
D. Crowdy and M. Cloke, \emph{Analytical solutions for distributed multipolar vortex equilibria
on a sphere}, Physics of Fluids, \textbf{15} (2003), 22--34.

\bibitem{DS22}
 A. Dea\~{n}o and N. Simm, \emph{Characteristic polynomials of complex random matrices and Painlev\'e transcendents}, Int. Math. Res.
Not., \textbf{2022} (2022), 210--264.

\bibitem{EM07}
P. Etingof and X. Ma, \emph{Density of eigenvalues of random normal matrices with an arbitrary potential, and of generalized normal matrices}, SIGMA \textbf{3} (2007), 048.

    \bibitem{FF11}
J.~Fischmann and P.J. Forrester, \emph{One-component plasma on a spherical
  annulus and a random matrix ensemble}, J. Stat. Mech. \textbf{2011} (2011),
  P10003.


    \bibitem{Fo10}
P.J. Forrester, \emph{Log-gases and random matrices}, Princeton University Press,
  Princeton, NJ, 2010.
  
   \bibitem{Fo12}
P.J. Forrester, \emph{Large deviation eigenvalue density for the soft edge Laguerre and Jacobi $\beta$-ensembles}, J.~Phys.~A
  {\bf 45} (2012), 145201.

  
\bibitem{Fo24+}
  P.J. Forrester, \emph{Dualities  in random matrix theory}, preprint.

      \bibitem{FK09}
P.J. Forrester and M.~Krishnapur, \emph{Derivation of an eigenvalue probability
  density function relating to the {P}oincar\'e disk}, J. Phys. A \textbf{42} (2009), 385204. 
  
  \bibitem{FR09}
P.J. Forrester and E.M. Rains, \emph{Matrix averages relating to the {G}inibre
  ensemble}, J. Phys. A \textbf{42} (2009), 385205
  
   \bibitem{FW15}
P.J. Forrester and N.S. Witte,  \emph{Painlev\'e II in random matrix theory and related fields}, Constr. Approx. \textbf{41}
 (2015),  589--613.


 \bibitem{GT11}
 B. Gustafsson and V.G. Tkachev, \emph{On the exponential transform of lemniscates},
Comp. Methods Function Theory \textbf{11} (2011), 591--615.
 
 \bibitem{HF12} 
 D. Holcomb and G.R.M. Flores, \emph{Edge Scaling of the $\beta$-Jacobi ensemble}, J. Stat. Phys. \textbf{149} (2012), 1136--1160.

 \bibitem{JT98}
B.~Jancovici and G.~T\'ellez, \emph{Two-dimensional {Coulomb} systems on a
  surface of constant negative curvature}, J. Stat. Phys. \textbf{91} (1998),
  953--977.


 \bibitem{KKL24}
 M. Kieburg, A.B.J.~Kuijlaars and S. Lahiry, \emph{Orthogonal polynomials in the normal matrix model with two insertions}, arXiv:2408.12952.

\bibitem{KLY23}
 T. Kr\"uger, S.-Y. Lee and M. Yang, \emph{Local statistics in normal matrix models with merging singularity}, arXiv:2306.12263.

   \bibitem{LY23}
S.-Y. Lee and M. Yang, \emph{Strong asymptotics of planar orthogonal polynomials: Gaussian weight perturbed by finite number of
point charges}, Comm. Pure Appl. Math. \textbf{76} (2023), 2888--2956.
  
  \bibitem{LD21} 
  A.R. Legg and P.D. Dragnev, \emph{Logarithmic equilibrium on the sphere in the presence of multiple point charges}, 
Constr. Approx., \textbf{54} (2021), 237--257.

\bibitem{MS14} S. N. Majumdar and G. Schehr, \emph{Top eigenvalue of a random matrix: Large deviations and third order phase transition}, J. Stat. Mech. \textbf{2014} (2014), P01012.

\bibitem{May13} A. Mays, \emph{A real quaternion spherical ensemble of random matrices}, J. Stat. Phys. \textbf{153} (2013), 48--69.

\bibitem{MP17}  A. Mays and A. Ponsaing, \emph{An induced real quaternion spherical ensemble of random matrices}, Random Matrices Theory Appl. \textbf{6} (2017), 1750001.

\bibitem{MSV79} D. S. Moak, E. B. Saff and R. S. Varga, \emph{On the zeros of Jacobi polynomials $P_n^{(\alpha_n,\beta_n)}(x)$}, Trans. Am. Math. Soc. \textbf{249} (1979), 159--162. 

\bibitem{RKC12} H.M. Ramli, E. Katzav and I.P. Castillo, \emph{Spectral properties of the Jacobi ensembles via the
Coulomb gas approach}, J. Phys. A: Math. Theor. \textbf{45} (2012), 465005.


 \bibitem{Ri72}
S. Richardson, \emph{Hele-Shaw 
flows with a free boundary produced by the injection of 
fluid into a
narrow channel}, J. Fluid Mech., \textbf{56} (1972), 609--618.
  
  \bibitem{ST97}
E.B.~Saff and V.~Totik , \emph{Logarithmic potentials with external fields}, Grundlehren der Mathematischen Wissenschaften [Fundamental Principles of Mathematical Sciences], Vol. 316, Springer-Verlag, Berlin, 1997, appendix B by Thomas Bloom. 

\bibitem{Se23}
A.~Serebryakov, \emph{Multi-point correlators in non-Hermitian matrices and beyond},
Ph.D. Thesis, University of Sussex, 2023.

\bibitem{SSD23}
A. Serebryakov, N. Simm, and G. Dubach, \emph{Characteristic polynomials of random truncations: Mo-
ments, duality and asymptotics}, Random Matrices: Theory and Applications, \textbf{12} (2023),  2250049.

\bibitem{SS24}
A. Serebryakov and N. Simm, \emph{Schur function expansion in non-Hermitian ensembles and averages of characteristic polynomials}, \\
 Ann. Henri Poincar\'e (2024). https://doi.org/10.1007/s00023-024-01483-6

  \bibitem{Se24}
  S.~Serfaty, \emph{Lectures on Coulomb and Riesz gases}, arXiv:2407.21194.


  
    \bibitem{TW94a}
C.A. Tracy and H.~Widom, \emph{Level-spacing distributions and the {Airy} kernel}, Commun.
  Math. Phys. \textbf{159} (1994), 151--174.


       \bibitem{Wa78}
K.W. Wachter, \emph{The strong limits of random matrix spectra for sample
  matrices of independent elements}, Annal. Prob. \textbf{6} (1978), 1--18.



   \bibitem{WW19} 
  C. Webb and M.D. Wong, \emph{On the moments of the characteristic polynomial of a
Ginibre random matrix}, Proc. Lond. Math. Soc., \textbf{118} (2019), 1017--1056.

   \bibitem{ZS99}
K.~Zyczkowski and H.-J. Sommers, \emph{Truncations of random unitary matrices},
  J. Phys. A \textbf{33} (2000), 2045--2057.

   
   \end{thebibliography}
\end{document}